\documentclass[11pt,a4paper]{article}

\usepackage{amssymb,amsmath,amsthm,graphics,graphicx,color}

\newcommand{\NP}{{\sf NP}}

\newcommand{\W}{{\sf W}}

\newtheorem{lemma}{Lemma}
\newtheorem{theorem}{Theorem}
\newtheorem{corollary}{Corollary}

\newcommand{\dist}{{\rm dist}}
\newcommand{\diam}{{\rm diam}}
\newcommand{\chord}{{\rm chord}}

\begin{document}

\title{Hadwiger Number of Graphs with Small Chordality\thanks{The research leading to these results has received funding from the Research Council of Norway and the European Research Council under the European Union's Seventh Framework Programme (FP/2007-2013) / ERC Grant Agreement n. 267959. A preliminary version of this paper appeared as an
extended abstract in the proceedings of WG 2014.}}

\author{
Petr A. Golovach\thanks{Department of Informatics, University of Bergen, Norway, e-mail: \texttt{\{petr.golovach,pinar.heggernes,pim.vanthof\}@ii.uib.no\}@ii.uib.no}} 
\addtocounter{footnote}{-1}
\and
Pinar Heggernes\footnotemark
\addtocounter{footnote}{-1}
\and
Pim van 't Hof\footnotemark
\and
Christophe Paul\thanks{CNRS, LIRMM, Montpellier France,  e-mail: \texttt{paul@lirmm.fr}}
}

\date{}

\maketitle

\begin{abstract}
The Hadwiger number of a graph $G$ is the largest integer~$h$ such that~$G$ has the complete graph $K_h$ as a minor. We show that the problem of determining the Hadwiger number of a graph is \NP-hard on co-bipartite graphs, but can be solved in polynomial time on cographs and on bipartite permutation graphs. We also consider a natural generalization of this problem that asks for the largest integer~$h$ such that~$G$ has a minor with~$h$ vertices and diameter at most $s$. We show that this problem can be solved in polynomial time on AT-free graphs when $s\geq 2$, but is \NP-hard on chordal graphs for every fixed $s\geq 2$. 
\end{abstract}

\section{Introduction}
\label{sec:intro}

The Hadwiger number of a graph $G$, denoted by $h(G)$, is the largest integer $h$ such that the complete graph $K_h$ is a minor of $G$. The Hadwiger number has been the subject of intensive study, not in the least due to a famous conjecture by Hugo Hadwiger from 1943~\cite{Hadwiger43} stating that the Hadwiger number of any graph is greater than or equal to its chromatic number. In a 1980 paper, Bollob\'{a}s, Catlin, and Erd\H{o}s~\cite{BCE80} called Hadwiger's conjecture ``one of the deepest unsolved problems in graph theory.'' Despite many partial results
the conjecture remains wide open more than 70 years after it first appeared in the literature. 

Given the vast amount of graph-theoretic results involving the Hadwiger number, it is natural to study the computational complexity of the  {\sc Hadwiger Number} problem, which is to decide, given an $n$-vertex graph $G$ and an integer $h$, whether the Hadwiger number of $G$ is greater than or equal to~$h$ (or, equivalently, whether $G$ has $K_h$ as a minor).
Rather surprisingly, it was not until 2009 that this problem was shown to be \NP-complete by Eppstein~\cite{Eppstein09}. Two years earlier, Alon, Lingas, and Wahl\'{e}n~\cite{AlonLW07} observed that the problem is fixed-parameter tractable when parameterized by $h$ due to deep results by Robertson and Seymour~\cite{RobertsonS95}. This shows that the problem of determining the Hadwiger number of a graph is in some sense easier than the closely related problem of determining the clique number of a graph, as the decision version of the latter problem is \W[1]-hard when parameterized by the size of the clique. Alon et al.~\cite{AlonLW07} showed that the same holds from an approximation point of view: they provided a polynomial-time approximation algorithm for the {\sc Hadwiger Number} problem with approximation ratio $O(\sqrt{n})$, contrasting the fact that it is \NP-hard to approximate the clique number of an $n$-vertex graph in polynomial time to within a factor better than $n^{1-\epsilon}$ for any $\epsilon>0$~\cite{Zuckerman06}.

Bollob\'{a}s, Catlin, and Erd\H{o}s~\cite{BCE80} referred to the Hadwiger number as the {\em contraction clique number}. This is motivated by the observation that for any integer $h$, a connected graph $G$ has $K_h$ as a minor if and only if 
$G$ has $K_h$ as a contraction. In this context, it is worth mentioning another problem that has recently attracted some attention from the parameterized complexity community. The {\sc Clique Contraction} problem takes as input an $n$-vertex graph $G$ and an integer $k$, and asks whether $G$ can be modified into a complete graph by a sequence of at most $k$ edge contractions. Since every edge contraction reduces the number of vertices by exactly~$1$, it holds that $(G,k)$ is a yes-instance of the {\sc Clique Contraction} problem if and only if $G$ has the complete graph $K_{n-k}$ as a contraction (or, equivalently, as a minor). Therefore, the {\sc Clique Contraction} problem can be seen as the parametric dual of the {\sc Hadwiger Number} problem, and is \NP-complete on general graphs. When parameterized by~$k$, the {\sc Clique Contraction} problem was recently shown to be fixed-parameter tractable~\cite{CaiG13,LokshtanovMS13}, but the problem does not admit a polynomial kernel unless \NP~$\subseteq$~co\NP/
poly~\cite{CaiG13}.

\medskip

In this paper, we study the computational complexity of the {\sc Hadwiger Number} problem on several graph classes of bounded chordality. For chordal graphs, which form an important subclass of 4-chordal graphs, the {\sc Hadwiger Number} problem is easily seen to be equivalent to the problem of finding a maximum clique, and can therefore be solved in linear time on this class~\cite{TY84}. In Section~\ref{sec:main}, we present polynomial-time algorithms for solving the {\sc Hadwiger Number} problem on two other well-known subclasses of $4$-chordal graphs: cographs and bipartite permutation graphs. We also prove that the problem remains \NP-complete on co-bipartite graphs, and hence on $4$-chordal graphs. The latter result implies that the problem is also \NP-complete on AT-free graphs, a common superclass of cographs and bipartite permutation graphs.

In Section~\ref{sec:diam}, we consider a natural generalization of the {\sc Hadwiger Number} problem, and provide additional results about finding large minors of bounded diameter.
We show that the problem of determining the largest integer~$h$ such that a graph $G$ has a minor with $h$ vertices and diameter at most~$s$ can be solved in polynomial time on AT-free graphs if $s\geq 2$. In contrast, we show that this problem is \NP-hard on chordal graphs for every fixed $s\geq 2$, and remains \NP-hard for $s=2$ even when restricted to split graphs. Observe that when $s=1$, the problem is equivalent to the {\sc Hadwiger Number} problem and thus \NP-hard on AT-free graphs and linear-time solvable on chordal graphs due to our aforementioned results.

\section{Preliminaries}\label{sec:defs}
We consider finite undirected graphs without loops or multiple edges. 
For each of the graph problems considered in this paper, we let $n=|V(G)|$ and $m=|E(G)|$ denote the number of vertices and edges, respectively, of the input graph $G$.
For a graph $G$ and a subset $U\subseteq V(G)$ of vertices, we write $G[U]$ to denote the subgraph of $G$ induced by $U$. We write $G-U$ to denote the subgraph of $G$ induced by $V(G)\setminus U$, and $G-u$ if $U=\{u\}$.
For a vertex $v$, we denote by $N_G(v)$ the set of vertices that are adjacent to $v$ in $G$.
The \emph{distance} $\dist_G(u,v)$ between vertices $u$ and $v$ of $G$ is the number of edges on a shortest path between them.
The \emph{diameter} $\diam(G)$ of $G$ is $\max\{\dist_G(u,v) \mid u,v\in V(G)\}$.
The \emph{complement} of $G$ is the graph $\overline{G}$ with vertex set $V(G)$, where two distinct vertices are adjacent in $\overline{G}$ if and only if they are not adjacent in $G$.
For two disjoint vertex sets $X,Y\subseteq V(G)$, we say that $X$ and $Y$ are \emph{adjacent} if there are $x\in X$ and $y\in Y$ that are adjacent in $G$.

We say that $P$ is a {\em $(u,v)$-path} if $P$ is a path that joins $u$ and $v$. The vertices of $P$ different from $u$ and $v$ are the \emph{inner} vertices of $P$.
We denote by $P_n$ and $C_n$ the path and the cycle on $n$ vertices respectively.
The {\em length} of a path is the number of edges in the path.
A set of pairwise adjacent vertices is a \emph{clique}.  
A \emph{matching} is a set $M$ of edges such that no two edges in $M$ share an end-vertex. A vertex incident to an edge of a matching $M$ is said to be \emph{saturated} by $M$.
We write $K_n$ to denote the \emph{complete} graph on $n$ vertices, i.e., graph whose vertex set is a clique. 
For two integers $a\leq b$, the \emph{(integer) interval} $[a,b]$ is defined as $[a,b]=\{i\in \mathbb{Z}\mid a\leq i\leq b\}$.  
If $a>b$, then $[a,b]=\emptyset$.

The \emph{chordality}  $\chord(G)$ of a graph $G$ is the length of a longest
induced cycle in $G$; if $G$ has no cycles, then $\chord(G)=0$.
For a non-negative integer $k$, a graph $G$ is \emph{$k$-chordal} if $\chord(G)\leq k$.
A graph is \emph{chordal} if it is $3$-chordal.
A graph is \emph{chordal bipartite} if it is both $4$-chordal and bipartite.
A graph is a \emph{split} graph if its vertex set can be partitioned in an independent set and a clique.
For a graph $F$, we say that a graph $G$ is \emph{$F$-free} if $G$ does not contain $F$ as an induced subgraph.
A graph is a \emph{cograph} if it is $P_4$-free.
Let $\sigma$ be a permutation of $\{1,\ldots,n\}$. A graph $G$ is said to be a \emph{permutation graph for $\sigma$} if $G$ has vertex set $\{1,\ldots,n\}$ and two vertices $i,j$ are adjacent if and only if $i,j$ are reversed by the permutation. A graph $G$ is a \emph{permutation} graph if $G$ is a permutation graph for some $\sigma$. A graph is a \emph{bipartite permutation} graph if it is bipartite and permutation.
An {\it asteroidal triple (AT)} is a set of three non-adjacent vertices such that between each pair of them there is a path that does not contain a neighbor of the third.
A graph is \emph {AT-free} if it contains no AT. 
Each of the above-mentioned graph classes can be recognized in polynomial (in most cases linear) time, and they are closed under taking induced subgraphs \cite{BrandstadtLS1999,Golumbic04}. See the monographs by Brandst{\"a}dt et al.~\cite{BrandstadtLS1999} and Golumbic \cite{Golumbic04} for more properties and characterizations of these classes and their inclusion relationships.

\medskip
\noindent
{\bf Minors, Induced Minors, and Contractions}.
Let $G$ be a graph and let $e\in E(G)$. The {\it contraction} of $e$ removes 
both end-vertices of $e$ and replaces them by a
new vertex adjacent to precisely those vertices to
which the two end-vertices were adjacent. We denote by $G/e$ the graph obtained from $G$ be the contraction of $e$. For a set of edges $S$, $G/S$ is the graph obtained from $G$ by the contraction of all edges of $S$. A graph $H$ is a \emph{contraction} of $G$ if $H=G/S$ for some $S\subseteq E(G)$. We say that $G$ is \emph{$k$-contractible} to $H$ if $H=G/S$ for some set $S\subseteq E(G)$ with $|S|\leq k$. 
A graph $H$ is an \emph{induced minor} of $G$ if a $H$ is a contraction of an induced subgraph of $G$.
Equivalently, $H$ is an induced minor of $G$ if $H$ can be obtained from $G$ by a sequence of vertex deletions and edge contractions. 
A graph $H$ is a {\em minor} of a graph $G$ if $H$ is a contraction of a subgraph of~$G$.
Equivalently, $H$ is a minor of $G$ if $H$ can be obtained from $G$ by a sequence of vertex deletions, edge deletions, and edge contractions. 

Let $G$ and $H$ be two graphs. An {\em $H$-witness structure} ${\cal W}$ of $G$ is a partition $\{W(x) \mid x\in V(H)\}$ of the vertex set of a (not necessarily proper) subgraph of $G$ into $|V(H)|$ sets called {\em bags}, such that the following two conditions hold:
\begin{itemize}
\item[(i)] each bag $W(x)$ induces a connected subgraph of $G$;
\item[(ii)] for all $x,y\in V(H)$ with $xy\in E(H)$, bags $W(x)$ and $W(y)$ are adjacent in~$G$.
\end{itemize}

In addition, we may require an $H$-witness structure to satisfy one or both of the following additional conditions:
\begin{itemize}
\item[(iii)] for all $x,y\in V(H)$ with $xy\notin E(H)$, bags $W(x)$ and $W(y)$ are not adjacent in~$G$;
\item[(iv)] every vertex of $G$ belongs to some bag.
\end{itemize}

By contracting each of the bags into a single vertex we observe that $H$ is a contraction, an induced minor, or a minor of $G$ if and only if $G$ has an $H$-witness structure ${\cal W}$ that satisfies conditions (i)--(iv), (i)--(iii), or (i)--(ii), respectively. We will refer to such a structure ${\cal W}$ as an {\em $H$-contraction structure}, an {\em $H$-induced minor structure}, and an {\em $H$-minor structure}, respectively. Observe that, in general, such a structure ${\cal W}$ is not uniquely defined. 

Let ${\cal W}$ be an $H$-witness structure of $G$, and let $W(x)$ be a bag of ${\cal W}$. We say that $W(x)$ is a \emph{singleton} if $|W(x)|=1$ and $W(x)$ is an \emph{edge-bag} if $|W(x)|=2$.
We say that $W(x)$ is a \emph{big bag} if $|W(x)|\geq 2$.

We conclude this section by presenting four structural lemmas that will be used in the polynomial-time algorithms presented in Section~\ref{sec:main}. The first lemma is due to Heggernes et al.~\cite{HHLLP14}. 

\begin{lemma}[\cite{HHLLP14}]
\label{l-big}
If a graph $G$ is $k$-contractible to a graph $H$, then any $H$-contraction structure ${\cal W}$ of $G$ satisfies the following properties: 
\begin{itemize}
\item ${\cal W}$ has at most $k$ big bags;
\item each bag of ${\cal W}$ contains at most $k+1$ vertices;
\item all the big bags of ${\cal W}$ together contain at most $2k$ vertices.
\end{itemize}
\end{lemma}

The next lemma readily follows from the definitions of a minor, an induced minor, and a contraction.

\begin{lemma}
\label{lem:clique}
For every connected graph $G$ and non-negative integer~$p$, the following statements are equivalent:
\begin{itemize}
\item $G$ has $K_p$ as a contraction;
\item $G$ has $K_p$ as an induced minor;
\item $G$ has $K_p$ as a minor.
\end{itemize}
\end{lemma}

We say that an $H$-induced minor structure ${\cal W}=\{W(x)\mid x\in V(H)\}$ is \emph{minimal} if there is no $H$-induced minor structure ${\cal W}'=\{W'(x)\mid x\in V(H)\}$ with $W'(x)\subseteq W(x)$ for every $x\in V(H)$ such that at least one inclusion is proper.

\begin{lemma}\label{lem:diam-bag}
For any minimal $K_p$-induced minor structure of a graph $G$, each bag induces a subgraph of diameter at most~$\max\{\chord(G)-3,0\}$.
\end{lemma}

\begin{proof}
Let ${\cal W}$ be a minimal $K_p$-induced minor structure of $G$. Since $K_p$ is a complete graph, any two bags of ${\cal W}$ are adjacent.
Let $W(x)$ be a bag. If $|W(x)|=1$, then $\diam(G[W(x)])=0$ and the statement holds. 
Observe that if $p\leq 2$, then each bag of $\cal W$ is a singleton
due to the minimality of ${\cal W}$. 
Suppose that $|W(x)|\geq 2$ and $p\geq 3$. 

Let $u,v\in W(x)$ be vertices such that $\dist_{G[W(x)]}(u,v)=\diam(G[W(x)])$.  Let $P$ be a shortest $(u,v)$-path in $G[W(x)]$.
Observe that by the choice of $u$ and $v$, the graphs $G[W(x)]-u$ and $G[W(x)]-v$ are connected. Because $\cal W$ is minimal, there are two bags $W(y)$ and $W(z)$ for distinct $y,z\in V(H)$ such that  $u$ is the unique vertex of $W(x)$ adjacent to a vertex of $W(y)$ and $v$ is the unique vertex of $W(x)$ adjacent to a vertex of $W(z)$. Because the graphs $G[W(y)]$ and $G[W(z)]$ are connected and the sets $W(y)$ and $W(z)$ are adjacent, there is an induced $(u,v)$-path $P'$ in $G$
whose inner vertices all belong to $W(y)\cup W(z)$.
Notice that $P'$ has length at least 3 and no inner vertex of $P'$ is adjacent to a vertex of $W(x)$ in $G$. Consequently, the union of $P$ and $P'$ is an induced cycle of length at least $\diam(G[W(x)])+3\leq \chord(G)$, so we conclude that $\diam(G[W(x)])\leq \chord(G)-3$.
\end{proof}

Note that Lemma~\ref{lem:diam-bag} immediately implies the aforementioned equivalence on chordal graphs between the {\sc Hadwiger Number} problem and the problem of finding a maximum clique. Lemma~\ref{lem:diam-bag} also implies the following result.

\begin{corollary}\label{cor:diam-bag} 
If $G$ is a graph of chordality at most~$4$, then for any minimal $K_p$-induced minor structure in $G$, each bag is a clique. 
\end{corollary}

We say that a $K_p$-induced minor structure is \emph{nice} if each bag is either a singleton or an edge-bag. 

\begin{lemma}
\label{lem:prism}
Let $G$ be a $\overline{C}_6$-free graph of chordality at most~$4$. If $K_p$ is an induced minor of $G$, then $G$ has a nice $K_p$-induced minor structure.
\end{lemma}

\begin{proof}
Let ${\cal W}$ be a minimal $K_p$-induced minor structure in $G$. By Corollary~\ref{cor:diam-bag}, each bag of ${\cal W}$ is a clique. Hence, in order to prove the lemma, it suffices to show that each bag contains at most two vertices.

For contradiction, suppose there exists a bag $W(x)$ that contains at least three distinct vertices $u_1,u_2,u_3$. Notice that because $W(x)$ is a clique, for any $u_i\in W(x)$, $G[W(x)]-u_i$ is connected. Because $\cal W$ is minimal, there are 3 bags $W(y_1)$, $W(y_2)$ and $W(y_3)$ for distinct $y_1,y_2,y_3\in V(H)$ such that  $u_i$ is the unique vertex of $W(x)$ adjacent to a vertex of $W(y_i)$ for $i\in\{1,2,3\}$. For any distinct $i,j\in\{1,2,3\}$, there is an induced $(u_i,u_j)$-path $P_{ij}$ in $G$
whose inner vertices all belong to $W(y_i)\cup W(y_j)$, because the graphs $G[W(y_i)]$ and $G[W(y_j)]$ are connected and the sets $W(y_i)$ and $W(y_j)$ are adjacent. 
Since $W(x)$ is a clique and $G$ has chordality at most $4$, path $P_{ij}$ has length 3.

Let $P_{12}=u_1v_1w_2u_2$, $P_{23}=u_2v_2w_3u_3$ and $P_{31}=u_3v_3w_1u_1$.
We select the paths $P_{ij}$ for $i,j\in\{1,2,3\}$ in such a way that they have the maximum number of common edges. We claim that any two distinct paths have a common edge.

Notice that $v_i,w_i\in W(y_i)$ and if $v_i\neq w_i$, then $v_iw_i\in E(G)$ because each $W(y_i)$ is a clique by Corollary~\ref{cor:diam-bag}. Suppose that for some index $i\in\{1,2,3\}$, say for $i=1$, $v_i\neq w_i$. Consider the cycle $u_1u_2v_2w_3v_3w_1u_1$ if $v_3\neq w_3$ and the cycle  $u_1u_2v_2v_3w_1u_1$ if $v_3=w_3$. Because $G$ has chordality at most 4, these cycles are not induced. Since $u_1,u_2$ are not adjacent to $v_3,w_3$, 
it implies that $v_2w_1\in E(G)$.  Then the path $P_{12}$ could be replaced by $u_1w_1v_2u_2$ and we would get more common edges in the paths. Therefore, $v_i=w_i$ for $i\in\{1,2,3\}$. 

It remains to observe that $G[\{u_1,u_2,u_3,v_1,v_2,v_3\}]$ is isomorphic to $\overline{C}_6$.
\end{proof}

\section{Computing the Hadwiger Number}\label{sec:main}
In this section we show that {\sc Hadwiger Number} problem can be solved in polynomial time on cographs and bipartite permutation graphs. We complement these results by showing that the problem is \NP-complete on co-bipartite graphs, another well-known subclass of the class of $4$-chordal graphs.

\subsection{Hadwiger number of cographs}\label{sec:cographs}
We need some additional terminology.

Let $G_1$ and $G_2$ be two graphs with $V(G_1) \cap V(G_2) = \emptyset$. The {\it
(disjoint) union} of $G_1$ and $G_2$ is defined as 
$$G_1\oplus G_2=(V(G_1)\cup V(G_2),E(G_1)\cup E(G_2)),$$
and the {\it join} of $G_1$ and $G_2$ is defined as 
$$G_2\otimes G_2=(V(G_1)\cup V(G_2), E(G_1)\cup E(G_2)\cup \{uv\mid u\in V(G_1), v\in V(G_2)\}).$$
It is well-known (see, e.g., \cite{BrandstadtLS1999,Golumbic04}) that every cograph can be constructed  recursively from isolated vertices using these two operations. Equivalently, cographs can be defined as follows.
A \emph{cotree} $T$ of a cograph $G$ 
is a rooted tree
with two types of interior nodes: $0$-nodes (corresponding to disjoint unions) and $1$-nodes (corresponding to joins). 
The vertices of $G$ are assigned to the leaves of $T$ in a 
one-to-one manner. Two vertices $u$ and $v$ are adjacent 
in $G$ if and only if the lowest common ancestor of the leaves 
$u$ and $v$ in $T$ is a $1$-node. A graph is a cograph if
and only if it has a cotree~\cite{CorneilLS81}.
Cographs can be recognized and their corresponding cotrees
can be generated in linear time \cite{CorneilPS85,HabibP05}.

\begin{theorem}\label{thm:cograph}
The {\sc Hadwiger Number} problem can be solved in $O(n^3)$ time on cographs.
\end{theorem}

\begin{proof}
Let $G$ be a cograph on $n$ vertices. We may assume that $G$ is connected, as otherwise we can simply consider the connected components of $G$ one by one. 
By Lemma~\ref{lem:clique}, it is sufficient to find the maximum $p$ such that $K_p$ is an induced minor of $G$. Because cographs are $\overline{C}_6$-free, we can use Lemma~\ref{lem:prism}.

For a non-negative integer $r$,
denote by $c_r(G)$ the largest integer $p$ 
such that $G$ has a nice $K_p$-induced minor structure with exactly~$r$ edge-bags.
If $G$ has no such structure for any~$p$, then $c_r(G)=0$. Notice that $c_0(G)$ is the size of a maximum clique in $G$. 

If $G$ has one vertex, then $c_r(G)=1$ if $r=0$ and $c_r(G)=0$ otherwise. It is also straightforward to see that 
$$c_r(G_1\oplus G_2)=\max\{c_r(G_1),c_r(G_2)\},$$
for any two disjoint graphs $G_1$ and $G_2$. 

We need the following observation about joins.

\medskip
\noindent
{\bf Claim 1. }{\it Let $G_1$ and $G_2$ be disjoint graphs, and suppose that $G=G_1\otimes G_2$ has a nice $K_p$-induced minor structure with $r>0$ edge-bags.
Then $G$ has a $K_p$-induced minor structure ${\cal W}=\{W(x)\mid x\in V(K_p)\}$ with $r$ edge-bags such that either $|V(G_1)\cap W(x)|\leq 1$ for every edge-bag $W(x)\in{\cal W}$, or $|V(G_2)\cap W(x)|\leq 1$ for every edge-bag $W(x)\in{\cal W}$.
}

\begin{proof}[Proof of Claim~1] 
Let ${\cal W}=\{W(x)\mid x\in V(K_p)\}$  be a nice  $K_p$-induced minor structure in $G$ that has $r$ edge-bags.
Suppose that $W(x)=\{u_1,v_1\}$ for $u_1,v_1\in V(G_1)$ and $W(y)=\{u_2,v_2\}$ for $u_2,v_2\in V(G_2)$. Because $G=G_1\otimes G_2$, 
the set $\{u_1,v_1,u_2,v_2\}$ induces $K_4$. 
We replace $W(x)$ and $W(y)$ by 
$W'(x)=\{u_1,u_2\}$ and $W'(y)=\{v_1,v_2\}$. Because $u_1,v_1$ are adjacent to the vertices of $G_2$ and $u_2,v_2$ are adjacent to the vertices of $G_1$, the bags  $W'(x)$ and $W'(y)$ are adjacent to every bag $W(z)$ for $z\neq x,y$. Therefore, we obtain a new nice $K_p$-induced minor structure with $r$ edge-bags. By doing such replacement recursively, we obtain a  nice $K_p$-induced minor structure with the desired property. This completes the proof of Claim~1. 
\renewcommand{\qedsymbol}{$\diamond$}
\end{proof}

Now we obtain the formula for $c_r(G_1\otimes G_2)$.

\medskip
\noindent
{\bf Claim 2. }{\it Let $G_1$ and $G_2$ be disjoint graphs, $n_1=|V(G_1)|$ and $n_2=|V(G_2)|$, and let $r$ be a non-negative integer. 
For a non-negative integer $s\leq r$ and $i=1,2$, let
$$c_{s,i}=
\begin{cases}
s+\min\{c_{r-s}(G_i),n_i-r\}+\min\{n_{3-i}-s,c_0(G_{3-i})\} &\mbox{if } n_i-2r+s\geq 0,\\
0&\mbox{if }n_i-2r+s<0.
\end{cases}
$$
Then
$$c_r(G_1\otimes G_2)=\max\{c_{s,i}\mid 0\leq s\leq \min\{n_1,n_2,r\},1\leq i\leq 2\}.$$
}

\begin{proof}[Proof of Claim~2] 
Let $G=G_1\otimes G_2$. Let $p=c_r(G)>0$ and let ${\cal W}=\{W(x)\mid x\in V(K_p)\}$  be a nice  $K_p$-induced minor structure in $G$ that has $r$ edge-bags. 
Let ${\cal W}_1^{(1)}=\{W(x)\mid W(x)=\{u\},u\in V(G_1)\}$, ~${\cal W}_1^{(2)}=\{W(x)\mid W(x)=\{u,v\},u,v\in V(G_1)\}$,
${\cal W}_2^{(1)}=\{W(x)\mid W(x)=\{u\},u\in V(G_2)\}$, ${\cal W}_2^{(2)}=\{W(x)\mid W(x)=\{u,v\},u,v\in V(G_2)\}$,
and ${\cal W}_3=\{W(x)\mid W(x)=\{u,v\},u\in V(G_1),v\in V(G_2)\}$. 
By Claim~1,  we can assume that ${\cal W}_1^{(2)}=\emptyset$ or ${\cal W}_2^{(2)}=\emptyset$. 
Suppose that ${\cal W}_2^{(2)}=\emptyset$; the case ${\cal W}_1^{(1)}=\emptyset$ is symmetric. Let $s=|{\cal W}_3|$.
Clearly, $s\leq\min\{n_1,n_2,r\}$. The set ${\cal W}_1^{(2)}$ has $r-s$ edge-bags, and these bags are disjoint with the bags of ${\cal W}_3$. Because each bag of ${\cal W}_3$ has one vertex in $V(G_1)$, it holds that $2(r-s)+s\leq n_1$. 
Hence,  
$$c_{s,1}=s+\min\{c_{r-s}(G_1),n_1-r\}+\min\{n_2-s,c_0(G_{2})\}.$$
Observe that $|{\cal W}_1^{(1)}|+|{\cal W}_1^{(2)}|\leq c_{r-s}(G_1)$.
Also because the bags of ${\cal W}_1^{(1)}$,  ${\cal W}_1^{(2)}$ and ${\cal W}_3$ are disjoint sets, $|{\cal W}_1^{(1)}|+2|{\cal W}_1^{(2)}|+s\leq n_1$.
We have that  $n_1\geq |{\cal W}_1^{(1)}|+2|{\cal W}_1^{(2)}|+s=|{\cal W}_1^{(1)}|+|{\cal W}_1^{(2)}|+(r-s)+s$ and $|{\cal W}_1^{(1)}|+|{\cal W}_1^{(2)}|\leq n_1-r$.
Clearly, $|{\cal W}_2^{(1)}|\leq c_{0}(G_2)$, and because the bags of ${\cal W}_2^{(1)}$  are disjoint with the bags of ${\cal W}_3$,
$|{\cal W}_2^{(1)}|+s\leq n_2$. Therefore, 
$|{\cal W}_2^{(1)}|\leq \min\{n_2-s,c_{0}(G_2)\}$.
We have that $p=|{\cal W}_1^{(1)}|+|{\cal W}_1^{(2)}|+|{\cal W}_2^{(1)}|+|{\cal W}_3|\leq 
s+\min\{c_{r-s}(G_1),n_1-r\}+\min\{n_2-s,c_0(G_{2})\}=c_{s,1}$. 
We conclude that 
$$c_r(G)\leq \max\{c_{s,i}\mid 0\leq s\leq \min\{n_1,n_2,r\},1\leq i\leq 2\}. $$

For the other direction, let $0\leq s\leq \min\{n_1,n_2,r\}$ and assume that $c_{s,1}\geq c_{s,2}$ (the other case is symmetric). Let $K$ be a maximum clique in $G_2$. Recall that $c_0(G_2)=|K|$. 

If $c_{s,1}=0$, then it is trivial to see that $c_r(G)\geq c_{s,1}$. Assume that $c_{s,1}>0$.
Then  $n_1-2r+s\geq 0$. Let ${\cal W}_1$ be a nice $K_q$-induced minor structure in $G_1$ with $r-s$ edge-bags for $q=c_{r-s}(G_1)$. 

Suppose that $c_{r-s}(G_1)\leq n_1-r$, i.e., $\min\{c_{r-s}(G_1),n_1-r\}=c_{r-s}(G_1)$. 
 We select a set $R_1$ of $s$ vertices in $V(G_1)$ such that
the vertices of $R_1$ are not included in the bags of ${\cal W}_1$. 
We always can do it because the bags of ${\cal W}_1$ contain $c_{r-s}(G_1)+r-s\leq n_1-s$ vertices of $G_1$. 
 Let $R_2$ be a set of $s$ vertices in $V(G_2)$ such that $R_2\cap K$ has minimum size. 
We consider a nice $K_p$-induced minor structure in $G$ that has $s$ edge-bags containing pairs of vertices $\{u,v\}$ for $u\in R_1$ and $v\in R_2$,
$r-s$ edge-bags that are edge-bags of ${\cal W}_1$, the singletons of ${\cal W}_1$, and 
$|K\setminus R_2|$ singletons $\{v\}$ for $v\in K\setminus R_2$. Here $p=s+c_{r-s}(G_1)+|K\setminus R_2|$. 
If $c_0(G_2)\leq n_2-s$, then $|K\setminus R_2|=c_0(G_2)$. If $c_0(G_2)\geq n_2-s$, then $|K\setminus R_2|=n_2-s$.
Then  
$c_r(G)\geq p=s+c_{r-s}(G_1)+|K\setminus R_2|\geq s+\min\{c_{r-s}(G_1),n_1-r\}+\min\{n_2-s,c_0(G_2)\}$. 

Finally, assume that $c_{r-s}(G_1)> n_1-r$, i.e.,  $\min\{c_{r-s}(G_1),n_1-r\}=n_1-r$.
Let $S$ be the set of vertices $u\in V(G_1)$ such that $\{u\}$ is a singleton in ${\cal W}$. We select a set $R_1$ of $s$ vertices in $V(G_1)$ such that
the vertices of $R_1$ are not included in the edge-bags of ${\cal W}_1$ and $R_1\cap S$ has minimum size.
We can find such a set $R_1$ because  ${\cal W}_1$ has $2(r-s)\leq n_1-s$ vertices in the edge-bags. 
Because the total number of vertices of $G_1$ in the bags of ${\cal W}_1$ is $c_{r-s}(G_1)+r-s>n_1-s$, $R_1\cap S\neq\emptyset$.
Also, $|S\setminus R_1|=n_1-2(r-s)-s=n_1-2r+s$. 
 Let $R_2$ be a set of $s$ vertices in $V(G_2)$ such that $R_2\cap K$ has minimum size. 
We consider a nice $K_p$-induced minor structure in $G$ that has $s$ edge-bags containing pairs of vertices $\{u,v\}$ such that $u\in R_1$ and $v\in R_2$,
$r-s$ edge-bags that are edge-bags of ${\cal W}_1$, 
$|S\setminus R_1|$ singletons $\{u\}$ such that $\{u\}\in {\cal W}_1$ and $u\in S\setminus R_1$, and 
$|K\setminus R_2|$ singletons $\{v\}$ such that $v\in K\setminus R_2$. Here $p=r+|S\setminus R_1|+|K\setminus R_2|$. 
Then  
$c_r(G)\geq p=r+|S\setminus R_1|+|K\setminus R_2|=r+(n_1-2r+s)+|K\setminus R_2|=s+(n_1-r)+|K\setminus R_2|
\geq s+\min\{c_{r-s}(G_1),n_1-r\}+\min\{n_2-s,c_0(G_2)\}$.

In all cases $c_r(G)\geq c_{s,1}$. By our assumption,  $c_{s,1}\geq c_{s,2}$. Hence,
$$c_r(G)\geq \max\{c_{s,i}\mid 0\leq s\leq \min\{n_1,n_2,r\},i=1,2\}.$$
and this completes the proof of Claim~2.
\renewcommand{\qedsymbol}{$\diamond$}
\end{proof}

In order to find the maximum~$p$ such that $K_p$ is an induced minor of $G$, we first compute a cotree of $G$, which can be done in linear time~\cite{CorneilPS85,HabibP05}. We then compute $c_r(G)$ for all $r\in \{0,\ldots,n\}$ using the obtained formulas for $c_r(G_1\oplus G_2)$ and $c_r(G_1\otimes G_2)$ in $O(n^3)$ time. Let $p=\max_{0\leq r\leq n}c_r(G)$. It remains to observe that by Lemma~\ref{lem:prism}, $K_p$ is the complete graph of maximum size that is an induced minor of~$G$.
\end{proof}

\subsection{Hadwiger number of bipartite permutation graphs}\label{sec:bip-perm}
Let us for a moment consider the class of chordal bipartite graphs. Recall that these are exactly the bipartite graphs that have chordality at most~$4$. It is well-known that chordal bipartite graphs form a proper superclass of the class of bipartite permutation graphs. Since chordal bipartite graphs have chordality at most~$4$ and are $\overline{C}_6$-free due to the absence of triangles, we can apply Lemma~\ref{lem:prism} to this class. Let us additionally observe that the number of singletons in any $K_p$-induced minor structure of a bipartite graph is at most~$2$. 

The above observations allow us to reduce the {\sc Hadwiger Number} problem on chordal bipartite graphs to a special matching problem as follows. We say that a matching~$M$ in a graph $G$ is a \emph{clique-matching} if for any two distinct edges 
$e_1,e_2\in M$, there is an edge in $G$ between an end-vertex of $e_1$ and an end-vertex of $e_2$. Now consider the following decision problem:

\medskip
\indent {\sc Clique-Matching}\\
\indent {\it Instance:} \hspace*{-.01cm} A graph $G$ and a positive integer $k$.\\
\indent {\it Question:} Is there a clique-matching of size at least $k$ in $G$?

\begin{lemma}\label{lem:match} 
If the {\sc Clique-Matching} problem can be solved in $f(n,m)$ time on chordal bipartite graphs, then the {\sc Hadwiger Number} problem can be solved in $O((n+m)\cdot f(n,m))$ time on this graph class.
\end{lemma}

\begin{proof}
Suppose that {\sc Clique-Matching} can be solved in $f(n,m)$ time on chordal bipartite graphs with $n$ vertices and $m$ edges. 
Let $(G,p)$ be an instance of the {\sc Hadwiger Number} problem, where $G$ is a chordal bipartite graph with $n$ vertices and $m$ edges. We assume that $G$ is connected, as otherwise we can simply consider the connected components of $G$ one by one. Let $V_1,V_2$ be a bipartition of~$V(G)$.

By Lemma~\ref{lem:clique}, it holds that $(G,p)$ is a yes-instance if and only if $K_p$ is an induced minor of $G$. Moreover, as a result of Lemma~\ref{lem:prism}, $K_p$ is an induced minor of $G$ if and only if $G$ has a nice $K_p$-induced minor structure $\cal W$.
Hence it remains to prove that we can decide in $O((n+m)\cdot f(n,m))$ time whether such a structure ${\cal W}$ exists. 
Recall that any $K_p$-induced minor structure of $G$ has at most two singletons due to the fact that $G$ is bipartite, and that any two bags in such a structure are adjacent. 

It is straightforward to see that $G$ has a nice $K_p$-induced minor structure without singletons if and only if $G$ has a clique-matching of size at least $p$. Hence, the existence of such a structure can be checked by solving {\sc Clique-Matching} for the instance $(G,p)$ in $f(n,m)$ time.
 
To verify the existence of a nice $K_p$-induced minor structure with one singleton $\{u\}$, we do as follows for every possible candidate vertex $u\in V(G)$. If $u\in V_1$, then we construct the graph $G'=G[(V_1\setminus\{u\})\cup N_G(u)]$, and we construct $G'=G[(V_2\setminus\{u\})\cup N_G(u)]$ if $u\in V_2$. We then solve {\sc Clique-Matching} for the instance $(G',p-1)$.

Finally, to check whether $G$ has a nice $K_p$-induced minor structure with two singletons, we check all edges $uv\in E(G)$ with $u\in V_1$ and $v\in V_2$. For each such edge $uv$, we construct $G'=G[(N_G(u)\setminus\{v\})\cup(N_G(v)\setminus\{u\}) ]$ and solve {\sc Clique-Matching} for the instance $(G',p-2)$.

It is clear that the total running time is $O((n+m)\cdot f(n,m))$.
\end{proof}

We will use the following characterization of bipartite permutation graphs given by  Spinrad, Brandst{\"a}dt, and Stewart \cite{SpinradBS87} (see also~\cite{BrandstadtLS1999}). 
Let $G$ be a bipartite graph and let $V_1,V_2$ be a bipartition of $V(G)$. An ordering of vertices of $V_2$ has the \emph{adjacency property} if for every $u\in V_1$, $N_G(u)$ consists of vertices which are consecutive in the ordering of $V_2$. An ordering of vertices of $V_2$ has the \emph{enclosure property} if for every pair  of vertices $u,v\in V_1$ such that $N_G(u)\subseteq N_G(v)$, vertices in $N_G(v)\setminus N_G(u)$ occur consecutively in the ordering of $V_2$. 

\begin{lemma}[\cite{SpinradBS87}]\label{lem:ordering}
Let $G$ be a bipartite graph with bipartition $V_1,V_2$. The graph $G$ is a bipartite permutation graph if and only there is an ordering of $V_2$ that has the adjacency and enclosure properties. Moreover, bipartite permutation graphs can be recognized and the corresponding ordering of $V_2$ can be constructed in linear time. 
\end{lemma}

\begin{theorem}\label{thm:perm}
The {\sc Clique-Matching} problem can be solved in $O(mn^4)$ time on bipartite permutation graphs.
\end{theorem}

\begin{proof}
Let $G$ be a bipartite permutation graph and let $V_1,V_2$ be a bipartition of the vertex set. We assume without loss of generality that $G$ has no isolated vertices. Let $n_1=|V_1|$ and $n_2=|V_2|$. We present a dynamic programming algorithm for the problem. For simplicity, the algorithm we describe only finds the size of a maximum clique-matching $M$ in $G$, but the algorithm can be modified to find a corresponding clique-matching as well.

Our algorithm starts by constructing an ordering $\sigma_2$ of $V_2$ that has the adjacency and enclosure properties, which can be done in linear time due to Lemma~\ref{lem:ordering}.
From now on, we denote the vertices of $V_2$ by their respective rank in $\sigma_2$, that is $V_2=\{1,\ldots,n_2\}$. Observe that for every vertex $u\in V_1$, $N_G(u)$ forms an interval of $\sigma_2$. The {\em rightmost} (resp.~{\em leftmost}) neighbor of $u$ in $\sigma_2$ is the vertex of $N_G(u)$ which is the largest (resp.~smallest) in $\sigma_2$.

Let $uv\in E(G)$ with $u\in V_1$ and $v\in V_2$ be an edge in $G$ such that $uv$ belongs to some maximum clique-matching in $G$ and there is no $v'\in V_2$ with $v'<v$ such that $v'$ is saturated by a maximum clique-matching in~$G$. Our algorithm guesses the edge $uv$ by trying all different edges of~$G$. For each guess of $uv$, it does as follows.

By the definition of $uv$, we can safely delete all vertices $v'\in V_2$ with $v'<v$. To simplify notation, we assume without loss of generality that $v=1$, so $uv=u1$. Denote by $r$ the rightmost neighbor of $u$. Then, by the adjacency property of $\sigma_2$, we have that $N_G(u)=[1,r]$.

The algorithm now performs the following preprocessing procedure.

\begin{itemize}
\item Find the vertices $v_1,\ldots,v_l\in V_1\setminus\{u\}$ (decreasingly ordered with respect to their rightmost neighbor) such that $[1,r]\subseteq N_G(v_i)$. 
By consecutively checking the intervals $N_G(v_1),\ldots, N_G(v_l)$ and selecting the rightmost available (i.e., not selected before) vertex in the considered interval, find the maximum set $S=\{j_1,\ldots,j_{h}\}$ of integers such that $j_1>\ldots>j_{h}>r$ and $j_i\in N_G(v_i)$ for $i\in \{1,\ldots,h\}$. Delete $v_1,\ldots,v_h$ from $G$.
\item Find the vertices $x_1,\ldots,x_s\in V_1\setminus\{u\}$  (decreasingly ordered with respect to their rightmost neighbor)
such that $[1,2]\subseteq N_G(x_i)$.
\item Find the vertices $y_1,\ldots,y_t\in V_1$ (increasingly ordered with respect to their leftmost neighbor)
 such that $1\notin N_G(y_i)$ and $r\in N_G(y_i)$.
\item Delete the vertices $r+1,\ldots,n_2$ from $V_2$.
\end{itemize}
The structure of the neighborhoods of $u$, $x_1,\ldots,x_s$ and $y_1,\ldots,y_t$ after this preprocessing procedure is shown in Fig~\ref{fig:struct}. 

\begin{figure}[ht]
\centering\scalebox{0.65}{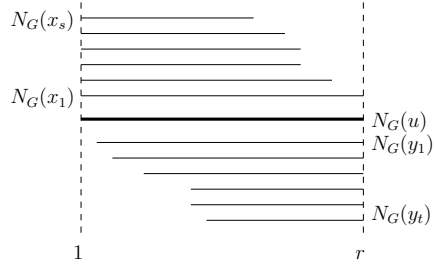}
\caption{Structure of the neighborhoods of $u$, $x_1,\ldots,x_s$ and $y_1,\ldots,y_t$ after the preprocessing procedure.
\label{fig:struct}}
\end{figure}

We prove that the preprocessing procedure is safe in the following claim.

\medskip
\noindent
{\bf  Claim~1.} {\it Let $M$ be a clique-matching of maximum size in $G$ such that $u1\in M$. Then there is a clique-matching $M'$ of maximum size such that $u1\in M'$ and 
\begin{itemize} 
\item[i)] $v_1j_1,\ldots,v_hj_h\in M'$,
\item[ii)] for any $vj\in M'$ such that $vj\neq u1$ and $v\notin \{v_1,\ldots v_h\}$, it holds that $v\in\{x_1,\ldots,x_s\}\cup\{y_1,\ldots,y_t\}$ and $j\in[2,r]$.
\end{itemize}
}

\begin{proof}[Proof of Claim~1]
Let $X=\{x\in V_1\mid [1,2]\subseteq N_G(x)\}$ and let $Y=\{y\in V_1\mid 1\notin N_G(y),r\in N_G(y)\}$. By the adjacency and enclosure property of $\sigma_2$, $X\cap Y=\emptyset$.

We show that for any $vj\in M$ such that $vj\neq u1$, $v\in X\cup Y$, and if $v\in Y$, then $j\leq r$. Because $u1\in M$, $j\geq 2$ and $v1\in E(G)$ or $uj\in E(G)$. If $v1\in E(G)$, then by the 
adjacency property, $[1,2]\subseteq N_G(v)$ and $v\in X$. If $uj\in E(G)$ and $1\in N_G(v)$, then $v\in X$ as well. Suppose that  $uj\in E(G)$ and $1\notin N_G(v)$. By the enclosure property applied to $u$ and $v$, we have that $r\in N_G(v)$ and $v\in Y$. Finally, because $uj\in E(G)$, $j\in N_G(u)$ and, therefore, $j\leq r$.

Let $p=\max\{j\mid j\in N_G(v), v\in X,p>r\}$ and $v\in X$ such that $p\in N_G(v)$. 
Notice that $p\geq j$ for any $xj\in M$.
 We prove that  there is a clique-matching $M'$ of maximum size such that $u1\in M'$ and $vp\in M'$. 

If $vp\in M$, then the statement trivially holds. Let $vp\notin M$. 
Assume that  $xp\notin M$ for any $x\in V_1$, i.e., assume $p$ is not saturated.  If $vj\in M$ for some $j\in V_2$, then we construct $M'$ by replacing $vj$ by $vp$ in $M$. If $vj\notin M$ for any $j\in V_2$, that is, if $v$ is not saturated, then we construct $M'$ by adding $vp$ to $M$. For any $xj\in M$, $1\leq j\leq p$ and, therefore, $j\in N_G(v)$. Hence, $vj\in E(G)$ and $M'$ is a clique-matching of maximum size.  
Suppose now that $xp\in M$ for some $x\in V_1$. Notice that $x\in X$ and $[1,p]\subseteq N_G(x)$. If  $vj'\in M$ for some $j'\in V_2$, then we construct $M'$ by replacing $vj',xp$ by $xj',vp$ in $M$. If $vj'\notin M$ for any $j'\in V_2$, then we construct $M'$ by replacing $xp$ by $vp$ in $M$.
Because $[1,p]\subseteq N_G(v)$,  $[1,p]\subseteq N_G(x)$ and for any $yj\in M$, $1\leq j\leq p$, we have that $M'$ is a clique-matching of maximum size.

We apply this statement inductively to obtain a clique-matching $M'$ of maximum size such that $u1\in M'$ and  $v_1j_1,\ldots,v_hj_h\in M'$. Assume that we have a maximum clique-matching $M'$ size such that $u1\in M'$ and  $v_1j_1,\ldots,v_ij_i\in M'$. We delete $v_1,\ldots,v_i$ and $j_1,\ldots,j_i$ from the graph. Then we find 
$p=\max\{j\mid j\in N_G(v), v\in X,p>r\}$. If such $p$ does not exist, we stop. Otherwise, we set $j_{i+1}=p$ and find $v_{i+1}\in X$ that is adjacent to $j_{i+1}$ and find a matching that contains $v_{i+1}j_{i+1}$.

To show ii), we observe that $\{x_1,\ldots,x_s\}=X\setminus \{v_1,\ldots,v_h\}$ and $\{y_1,\ldots,y_t\}=Y$. If $vj\in M'$ such that $v\notin \{v_1,\ldots v_h\}$ and $j>r$, then $x\in X$. But then at least one additional element should be included in the set $S$ constructed in the preprocessing procedure, contradicting its maximality. Hence, $j\in[2,r]$. This completes the proof of Claim~1.
\renewcommand{\qedsymbol}{$\diamond$}
\end{proof}

In the next stage of the algorithm we apply dynamic programming. For every $i\in\{0,\ldots,s\}$, $j\in\{0,\ldots,t\}$ and non-negative integer $\ell$, let $c(i,j,\ell)$ denote the size of a maximum clique-matching $M$ such that
\begin{itemize}
\item[a)] $u1\in M$,
\item[b)] for any $vp\in M$ such that $vp\neq u1$, it holds that $v\in\{x_1,\ldots,x_i\}\cup\{y_1,\ldots,y_j\}$, and
\item[c)] there are at most $\ell$ vertices in $[a_{i,j},b_{i,j}]=(\bigcap_{p=1}^i N_G(x_p))\cap(\bigcap_{q=1}^j N_G(y_q))$ saturated by $M$.
\end{itemize}

Recall that the vertices of $X$ and $Y$ are ordered with respect to their rightmost and leftmost neighbors, respectively. Hence,
for any $1\leq p<q\leq i$, we have $1\in N_G(x_q)\subseteq N_G(x_p)\subseteq [1,r]$, and for any $1\leq p<q\leq j$, we have $1\notin N_G(y_q)\subseteq N_G(y_p)\subseteq [2,r]$.
In particular, $[a_{i,j},b_{i,j}]=N_G(x_i)\cap N_G(y_j)$ for $i,j>0$. In other words, if $[a_{i,j},b_{i,j}]\neq \emptyset$, then $a_{i,j}$ is the left end-point of the interval $N_G(y_j)$ and
$b_{i,j}$ is the right end-point of the interval $N_G(x_j)$.
Observe that 
it can happen that $[a_{i,j},b_{i,j}]=\emptyset$. Observe also that $c(i,j,\ell)=c(i,j,b_{i,j}-a_{i,j}+1)$ if $[a_{i,j},b_{i,j}]\neq \emptyset$ and $\ell>b_{i,j}-a_{i,j}+1$. 
Hence, it is sufficient to compute $c(i,j,\ell)$ for $\ell\leq b_{i,j}-a_{i,j}+1\leq n_2$. 

Because all the vertices in $[a_{i,j},b_{i,j}]$ have the same neighbors in $\{x_1,\dots, x_i\}\cup\{y_1,\dots, y_j\}$, we can make the following observation.

\medskip
\noindent
{\bf Claim~2.} {\it Let $M$ be a clique-matching of maximum size such that $M$ satisfies a)--c) and $M$ has exactly $f$ saturated vertices in $[a_{i,j},b_{i,j}]$, and let $W\subseteq [a_{i,j},b_{i,j}]$ be a set of size $f$. Then there is a clique-matching $M'$ of maximum size that satisfies a)--c) such that $W$ is the set of vertices of $[a_{i,j},b_{i,j}]$ saturated by $M'$.
}
\medskip

If $i=j=0$, then we set $c(i,j,\ell)=1$ taking into account the matching with the unique edge $u1$. For other values of $i,j$, $c(i,j,\ell)$ is computed as follows.
To simplify notation, we assume that $x_0=y_0=u$.

\medskip
\noindent
{\bf Computation of $c(i,j,\ell)$ for $i>0,j=0$.} Because $1\in N_G(x_q)\subseteq N_G(x_p)\subseteq [1,r]$ for every $1\leq p<q\leq i$, any matching with edges incident to $x_1,\ldots,x_i$ is a clique-matching.  This observation also implies that a maximum matching can be obtained in greedy way.
 Notice that $[a_{i,0},b_{i,0}]=N_G(x_i)$.  By consecutively checking the intervals $N_G(x_1),\ldots, N_G(x_i)$ and selecting the rightmost available (i.e., not selected before) vertex in the considered interval, we find the maximum set $\{p_1,\ldots,p_{q}\}$ of integers such that $t\geq p_1>\ldots>p_{q}>1$, $p_f \in N_G(x_f)$ for $f\in \{1,\ldots,q\}$, and $|\{p_1,\ldots,p_q\}\cap [a_{i,0},b_{i,0}]|\leq \ell-1$.  
Taking into account the edge $u1$, we observe that $M=\{u1,x_1p_1,\ldots,x_qp_q\}$ is a required matching, and 
we have that $c(i,j,\ell)=q+1$.

\medskip
\noindent
{\bf Computation of $c(i,j,\ell)$ for $i=0,j>0$.} Now we have that $r\in N_G(y_q)\subseteq N_G(y_p)\subseteq [2,r]$ for every $1\leq p<q\leq j$. Hence, any matching with edges incident to $y_1,\ldots,y_j$ is a clique-matching and a maximum matching can be obtained in greedy way.
 Notice that $[a_{0,j},b_{0,j}]=N_G(y_j)$.  By consecutively checking the intervals $N_G(y_1),\ldots, N_G(y_j)$ and selecting the leftmost available (i.e., not selected before) vertex in the considered interval, we find the maximum set $\{p_1,\ldots,p_{q}\}$ of integers such that $1<p_1<\ldots<p_{q}\leq r$, $p_f \in N_G(y_f)$ for $f\in \{1,\ldots,q\}$, and $|\{p_1,\ldots,p_q\}\cap [a_{0,j},b_{0,j}]|\leq \ell$. It is straightforward to see that $M=\{u1,y_1p_1,\ldots,y_qp_q\}$ is a required matching, and we have that $c(i,j,\ell)=q+1$. 

\medskip
\noindent
{\bf Computation of $c(i,j,\ell)$ for $i>0,j>0$.} We compute  $c(i,j,\ell)$ using the tables of already computed values $c(i-1,j',\ell')$ for $j'\leq j$. We find the size of a maximum clique-matching $M$ by considering all possible choices for the vertex $x_i$ and then take the maximum among the obtained values. We distinguish three cases. Recall that $[a_{i,j},b_{i,j}]=N_G(x_i)\cap N_G(y_j)$.

\medskip
\noindent
{\em Case 1.} The vertex $x_i$ is not saturated by $M$. We have that $[a_{i-1,j},b_{i-1,j}]=N_G(x_{i-1})\cap N_G(y_j)\subseteq [a_{i,j},b_{i,j}]$ and $|[a_{i,j},b_{i,j}]\setminus [a_{i-1,j},b_{i-1,j}]|=b_{i-1,j}-b_{i,j}$. By Claim~2 implies that for any maximum clique-matching $M$ that satisfies a)--c) and has no edge incident to $x_i$, it holds that a clique-matching $M'$ of maximum size that satisfies a)--b), has no edge incident to $x_i$, and has at most $\ell'=\ell+b_{i-1,j}-b_{i,j}$ saturated vertices in $[a_{i-1,j},b_{i-1,j}]$ has the same size as $M$. Hence $c(i,j,\ell)=c(i-1,j,\ell')$.

\medskip
Now we consider the cases when $x_i$ is saturated by $M$. Denote by $p\in N_G(x_i)$ the vertex  such that $x_ip\in M$. 

\begin{figure}[ht]
\centering\scalebox{0.65}{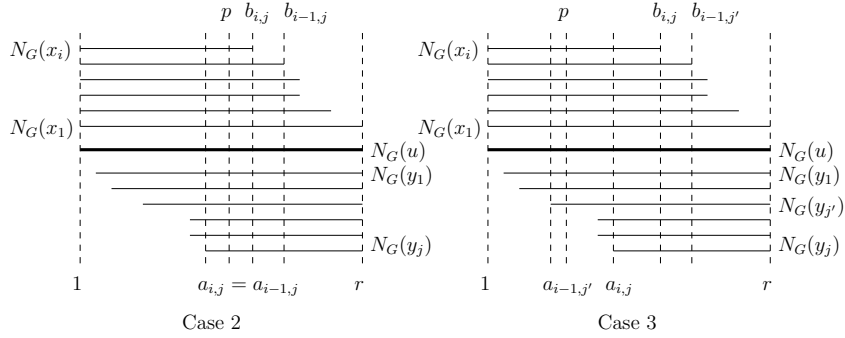}
\caption{Structure of the neighborhoods of $u$, $x_1,\ldots,x_i$ and $y_1,\ldots,y_j$ in Cases~2 and 3. 
\label{fig:cases}}
\end{figure}

\medskip
\noindent
{\em Case 2.} Vertex $p\in [a_{i,j},b_{i,j}]$ (see Fig.~\ref{fig:cases}). Observe that $p$ is adjacent to every vertex in $\{x_1,\ldots,x_{i-1}\}\cup \{y_1,\ldots,y_j\}$. Hence, for any edge $vq$ such that $v\in \{u\}\cup\{x_1,\ldots,x_{i-1}\}\cup\{y_1,\ldots,y_j\}$ and $q\neq p$, $x_ip$ and $vq$ have adjacent end-vertices, i.e., this choice of $p$ does not influence the selection of other edges of $M$ except that we can have at most $\ell-1$ other saturated vertices in $[a_{i,j},b_{i,j}]$. 
We have that $[a_{i-1,j},b_{i-1,j}]=N_G(x_{i-1})\cap N_G(y_j)\subseteq [a_{i,j},b_{i,j}]$ and $|[a_{i,j},b_{i,j}]\setminus [a_{i-1,j},b_{i-1,j}]|=b_{i-1,j}-b_{i,j}$. By Claim~2, we obtain that for any maximum clique-matching $M$ that satisfies a)--c) and $x_ip\in M$,  a clique-matching $M'$ of maximum size that satisfies a)--b), has no edge incident to $x_i$ and has at most $\ell'=\ell+b_{i-1,j}-b_{i,j}-1$ saturated vertices in $[a_{i-1,j},b_{i-1,j}]$ has the same size as $M$. Hence $c(i,j,\ell)=c(i-1,j,\ell')$.

\medskip
\noindent
{\em Case 3.} Vertex $p\notin [a_{i,j},b_{i,j}]$, i.e., $p<a_{i,j}$ (see Fig.~\ref{fig:cases}). Let $j'=\max\{f\mid p\in N_G(y_f),0\leq f\leq j\}$. As $p<a_{i,j}$, $j'<j$. 

Let $f\in\{j'+1,\ldots,j\}$, $g\in N_G(y_f)$ and $g>b_{i,j}$. Recall that $b_{i,j}$ is the right end-point of $N_G(x_i)$. Hence,  $x_ig\notin E(G)$. Because $f>j'$, $x_fp\notin E(G)$. We conclude that such edges cannot be in $M$. Similarly, let $f\in\{j'+1,\ldots,j\}$, $g\in N_G(y_f)$ and $g\leq b_{i,j}$. Then for any $v\in\{x_1,\ldots,x_i\}\cup\{y_1,\ldots,y_{j'}\}$, it holds that $vg\in E(G)$. Also if $j'+1\leq f<f'\leq j$, then for any $g\in N_G(x_{f'})$, $x_fg\in E(G)$. We have that it is safe to include in a clique-matching edges $x_fq$ for 
$f\in\{j'+1,\ldots,j\}$, $g\in N_G(y_f)$ and $g\leq b_{i,j}$. We select such edges in a greedy way.
By consecutively checking the intervals $N_G(y_{j'+1}),\ldots, N_G(y_j)$ and selecting the leftmost available (i.e., not selected before) vertex in the considered interval, we find the maximum set 
$\{g_1,\ldots,g_{q}\}$ of integers such that $p<g_1<\ldots<g_{q}\leq b_{i,j}$, $g_f \in N_G(y_{f+j'})$ for $f\in \{1,\ldots,q\}$ and $|\{g_1,\ldots,g_q\}\cap [a_{i,j},b_{i,j}]|\leq \ell$. 

\medskip
\noindent
{\bf Claim~3.} {\it Let $M$ be a clique-matching of maximum size that satisfies a)--c) and $x_ip\in M$.
 Then there is a clique-matching $M'$ of maximum size that satisfies a)--c) and $x_ip\in M'$ such that $y_{j'+1}g_1,\ldots,y_{j'+q}g_q\in M'$ and for any $vf\in M'$, it holds that 
$v\in\{y_{j'+1},\ldots,y_{j'+q}\}\cup\{x_1,\ldots,x_i\}\cup\{y_1,\ldots,y_{j'}\}$.
}

\begin{proof}[Proof of Claim~3]
We first show inductively that for every $f\in \{1,\ldots,q\}$,
there is a clique-matching $M'$ of maximum size that satisfies a)--c) and $x_ip\in M'$ such that $y_{j'+1}g_1,\ldots,y_{j'+f}g_f\in M'$.
Assume that  $y_{j'+1}g_1,\ldots,y_{j'+f-1}g_{f-1}\in M$.
If $y_{j'+f}g_{f}\in M$ then the statement trivially holds. Let $y_{j'+f}g_{f}\notin M$.
Suppose that $g_f$ is not saturated by $M$. If $y_{j'+f}h$ for some $j\in N_G(y_{j'+f})$, then we construct $M'$ by replacing $y_{j'+f}h$ by $y_{j'+f}g_f$ in $M$. 
If $y_{j'+f}$ is not saturated, then we construct $M'$ by adding $y_{j'+f}g_f$ to $M$. We already proved that it is safe to include $y_{j'+f}g_f$ in a clique-matching, i.e., the obtained matching $M'$ is a clique-matching.
Suppose now that $vg_f\in M$ for some $v\in \{x_1,\ldots,x_i\}\cup\{y_1,\ldots,y_j\}$. Notice that $v\notin \{y_{j'+1},\ldots,y_{j'+f-1}\}$ because 
$g_1<\ldots<g_{f-1}<g_f$. 
Suppose that  $y_{j'+f}h\in M$ for some $h\in N_G(y_{j'+f})$. Then $h\leq b_{i,j}$. By the selection of $g_1,\ldots,g_{f-1}$, it holds that $g_{f}<h$. 
It follows that $h$ is adjacent to $v$.
We construct $M'$ by replacing $vg_f,y_{j'+f}h$ by $y_{j'+f}g_f,vh$  in $M$. 
If $y_{j'+f}$ is not saturated, then we construct $M'$ by replacing $vg_f$ by $y_{j'+f}g_f$ in $M$. We again obtain a clique-matching.

We get a clique-matching $M'$ of maximum size that satisfies a)--c) and $x_ip\in M'$ such that $y_{j'+1}g_1,\ldots,y_{j'+q}g_q\in M'$.
It remains to show that  for any $vf\in M'$, it holds that 
$v\in\{y_{j'+1},\ldots,y_{j'+q}\}\cup\{x_1,\ldots,x_i\}\cup\{y_1,\ldots,y_{j'}\}$, i.e., $y_{j'+q+1},\ldots,y_j$ are not saturated. To see this, it is sufficient to observe that otherwise our greedy procedure would have added one more element to $\{g_1,\ldots,g_{q}\}$, contradicting the maximality of this set. This completes the proof of Claim~3.
\renewcommand{\qedsymbol}{$\diamond$}
\end{proof}

Observe that the total number of saturated vertices in 
$[a_{i-1,j'},b_{i-1,j'}]$ should be at most $(a_{i,j}-a_{i-1,j'})+(b_{i-1,j'}-b_{i,j})+\ell$.
Using Claims~2 and 3 and taking into account that $x_ip\in M$, we obtain that  
$c(i,j,\ell)=c(i-1,j',\ell')$ for $\ell'=(a_{i,j}-a_{i-1,j'})+(b_{i-1,j'}-b_{i,j})+\ell-(q+1)$.

\medskip
By our dynamic programming algorithm we eventually compute $c(s,t,\ell)$ for $\ell=0$ if $[a_{i,j},b_{i,j}]=\emptyset$ or $\ell=b_{i,j}-a_{i,j}+1$ if $[a_{i,j},b_{i,j}]\neq\emptyset$. Then
$c(s,t,\ell)$ is the size of a maximum clique-matching $M$ such that
\begin{itemize}
\item[a)] $u1\in M$,
\item[b)] for any $yp\in M$ such that $vp\neq u1$, it holds that $v\in\{x_1,\ldots,x_i\}\cup\{y_1,\ldots,y_j\}$.
\end{itemize}
By Claim~1, the size of a maximum clique-matching $M$ in $G$ such that $u1\in M$ is $c(s,t,\ell)+|S|$, where $S$ is the set of vertices constructed during the preprocessing procedure. 
Recall that the algorithm tries all possible choices for the edge $uv$, implying that our algorithm indeed computes the size of a maximum clique-matching in $G$.

It remains to evaluate the running time to complete the proof. Constructing the ordering $\sigma_2$ of $V_2$ can be done in $O(n+m)$ time by Lemma~\ref{lem:ordering}. The algorithm considers $m$ choices for the edge $uv$. For each of these choices, the preprocessing procedure can be performed in $O(n)$ time given the orderings of $V_1$ and $V_2$ (notice that Lemma~\ref{lem:ordering} is symmetric with respect to $V_1,V_2$, so we can obtain an ordering of $V_1$ with the adjacency and enclosure properties, too). Each step of the dynamic programming can be done in $O(n^2)$ time using the orderings of $V_1,V_2$. Observe that in this time we can compute $c(i,j,\ell)$ for all values of $\ell$. Hence, the dynamic programming algorithm runs in time $O(n^4)$. We conclude that the total running time is $O(mn^4)$.
\end{proof}

Combining Lemma~\ref{lem:match} and Theorem~\ref{thm:perm} yields the following result.

\begin{corollary}\label{cor:perm}
The {\sc Hadwiger Number} problem can be solved in $O((n+m)\cdot mn^4)$ time on bipartite permutation graphs.
\end{corollary}

\subsection{Hadwiger number of co-bipartite graphs }
\label{sec:cobip}

To conclude this section, we show that  the {\sc Hadwiger Number} problem is \NP-complete on co-bipartite graphs.

\begin{theorem}
\label{t-co-bipartite}
The {\sc Hadwiger Number} problem is \NP-complete on co-bipartite graphs.
\end{theorem}

\begin{proof}
First observe that, as a result of Lemma~\ref{lem:clique} and the observation that every edge contraction reduces the number of vertices by exactly~$1$, an $n$-vertex graph $G$ has $K_p$ as a minor if and only if $G$ is $(n-p)$-contractible to a complete graph. Hence, it suffices to prove that the {\sc Clique Contraction} problem is \NP-complete on co-bipartite graphs. In order to do so, we give a reduction from {\sc Not-All-Equal-3-SAT} (NAE-3-SAT), which is the problem of deciding, given boolean formula $\varphi$ in 3-CNF, whether there exists a satisfying truth assignment for $\varphi$ that does not set all the literals of any clause to true. Let $\varphi$ be an instance of this problem, and let $x_1,\ldots,x_n$ and $c_1,\ldots,c_m$ denote the variables and clauses of $\varphi$, respectively. 

We construct a graph $G$ as follows. For each $i\in \{1,\ldots,n\}$, we create two {\em variable vertices} $x_i$ and $\overline{x_i}$, as well as the edge $x_i\overline{x_i}$. Let $X= \{x_1,\overline{x_1},\ldots,x_n,\overline{x_n}\}$. For each clause $c_j$, we create $4n-3$ {\em clause vertices} $c_j^{1},\ldots,c_j^{4n-3}$ each of which is made adjacent to $x_i$ (respectively $\overline{x_i}$) if variable $x_i$ appears positively (respectively negatively) in clause $c_j$. For each $i\in \{1,\ldots,n\}$, we create $4n-3$ {\em dummy vertices} that are made adjacent to both $x_i$ and $\overline{x_i}$ but not adjacent to $x_j$ and $\overline{x_j}$ for every $j\neq i$. Finally, we add edges to make $X$ into a clique and to make $V(G)\setminus X$ into a clique. This completes the construction of $G$. 

Let $k=2n-2$ and $N=|V(G)\setminus X|=(4n-3)(n+m)$. Observe that $G$ is a co-bipartite graph on $2n+N$ vertices. We claim that $G$ is $k$-contractible to a complete graph if and only if $\varphi$ is a yes-instance of NAE-3-SAT. Note that by the definition of $k$ and $N$, graph $G$ is $k$-contractible to a complete graph if and only if $K_{N+2}$ is a contraction of $G$.

First suppose there exists a satisfying truth assignment $t$ for $\varphi$ that sets at least one literal to false in each clause. Let $W_0$ and $W_1$ denote the literals that $t$ sets to false and true, respectively. Let $G'$ denote the graph obtained from $G$ by contracting $W_i$ into a single vertex $w_i$, for $i\in \{0,1\}$. We claim that $G'$ is isomorphic to $K_{N+2}$. Observe that all the vertices of $V(G)\setminus X$ form a clique of size $N$ in $G$, and hence also in $G'$. Moreover, each of the dummy vertices is adjacent to both $w_0$ and $w_1$ due to the fact that $|W_i\cap \{x_j,\overline{x_j}\}|=1$ for every $i\in \{0,1\}$ and $j\in \{1,\ldots,n\}$. Finally, each of the clause vertices is adjacent to both $w_0$ and $w_1$, since $t$ sets at least one literal to true and at least one literal to false in each clause.

For the reverse direction, suppose $G$ has a $K_{N+2}$-contraction structure ${\cal W}$. Recall that for each $i\in \{1,\ldots,n\}$, there exist $4n-3=2k+1$ dummy vertices that are adjacent to both $x_i$ and $\overline{x_i}$, but to no other vertex in $X$. Hence Lemma~\ref{l-big} implies that for each $i\in \{1,\ldots,n\}$, there is a dummy vertex $d_i$ such that $\{d_i\}$ is a singleton of ${\cal W}$ and $N_G(d_i)\cap X=\{x_i,\overline{x_i}\}$. Using this, we now show that there are exactly two bags $W_0,W_1\in {\cal W}$ that are included in $X$.

For contradiction, first suppose that at most one bag of ${\cal W}$ does not contain a vertex from $V(G)\setminus X$. Then at least $N+1$ bags must contain a vertex of $V(G)\setminus X$. This is not possible, since $|V(G)\setminus X|=N$ and bags are disjoint by definition. Now suppose, again for contradiction, that there are three bags of ${\cal W}$ that do not intersect $V(G)\setminus X$. Then one of them, say $W$, contains neither $x_1$ nor $\overline{x_1}$. But then $W$ is not adjacent to the singleton $\{d_1\}$, contradicting the fact that ${\cal W}$ is a $K_p$-contraction structure of $G$. We conclude that there are exactly two bags $W_0,W_1\in {\cal W}$ that do not contain any vertex from $V(G)\setminus X$.

Since each of the singletons $\{d_1\},\ldots,\{d_n\}$ is adjacent to both $W_0$ and $W_1$, it holds that $|W_i\cap \{x_j,\overline{x_j}\}|=1$ for every $i\in \{0,1\}$ and $j\in \{1,\ldots,n\}$. Hence we can obtain a truth assignment $t$ for $\varphi$ by setting the literals in $W_0$ to false and the literals in $W_1$ to true. It remains to argue that for each clause $c_j$, at least one literal in $c_j$ is set to true and at least one literal is set to false by $t$. This follows from the fact that for every $j\in \{1,\ldots,m\}$, at least one of the $4n-3=2k+1$ clause vertices $c_j^{1},\ldots,c_j^{4n-3}$ forms a small bag due to Lemma~\ref{l-big}, and hence must be adjacent to both $W_0$ and $W_1$. This completes the proof.
\end{proof}

\section{Minors of Bounded Diameter}
\label{sec:diam}

In this section, we consider a generalization of the {\sc Hadwiger Number} problem where the aim is to obtain a minor of bounded diameter. 
Let $s$ be a positive integer. An {\em $s$-club} is a graph that has diameter at most~$s$. We consider the following problem:

\medskip
\noindent
\indent {\sc Maximum $s$-Club Minor}\\
\indent {\it Instance:} \hspace*{-.01cm} A graph $G$ and a non-negative integer $h$.\\
\indent {\it Question:} Does $G$ have a minor with $h$ vertices and diameter at most $s$?

\medskip
When $s=1$, the above problem is equivalent to the {\sc Hadwiger Number} problem. Recall that, due to Lemma~\ref{lem:clique}, the {\sc Hadwiger Number} problem can be seen as the parametric dual of the {\sc Clique Contraction} problem. The following straightforward lemma, which generalizes Lemma~\ref{lem:clique}, will allow us to formulate the parametric dual of the {\sc Maximum $s$-Club Minor} problem in a similar way.

\begin{lemma}\label{lem:diameterequiv}
For every connected graph $G$ and non-negative integers~$p$ and~$s$, the following statements are equivalent:
\begin{itemize}
\item $G$ has a graph with $p$ vertices and diameter at most $s$ as a contraction;
\item $G$ has a graph with $p$ vertices and diameter at most $s$ as an induced minor;
\item $G$ has a graph with $p$ vertices and diameter at most $s$ as a minor.
\end{itemize}
\end{lemma}

Lemma~\ref{lem:diameterequiv} implies that for any non-negative integer~$s$, the parametric dual of the {\sc Maximum $s$-Club Minor} problem can be formulated as follows:

\medskip
\noindent
\indent {\sc $s$-Club Contraction}\\
\indent {\it Instance:} \hspace*{-.01cm} A graph $G$ and a positive integer $k$.\\
\indent {\it Question:} Does there exist a graph $H$ with diameter at most~$s$ such that $G$\\
\hspace*{2.3cm} is $k$-contractible to $H$?

\medskip
\noindent
Observe that {\sc $1$-Club Contraction} is \NP-complete on AT-free graphs as a result of Theorem~\ref{t-co-bipartite}. We show that when $s\geq 2$, the problem becomes tractable on this graph class, even if $s$ is given as part of the input. On chordal graphs, the situation turns out to be opposite. Recall that the {\sc Hadwiger Number} problem, and hence the {\sc $1$-Club Contraction} problem, can be solved in linear time on chordal graphs. In contrast, we show that the $s$-{\sc Club Contraction} problem is \NP-complete on chordal graphs for every fixed $s\geq 2$, and the problem remains \NP-complete even when restricted to split graphs in case $s=2$.

\subsection{$s$-Club Contraction for AT-free graphs}\label{sec:AT}

We need some additional terminology and technical results.

For two paths $P=x_1\ldots x_s$ and $Q=y_1\ldots y_t$ such that $x_s=y_1$ and $V(P)\cap V(Q)=\{y_1\}$,  $P+Q$ is the \emph{concatenation} of $P$ and $Q$, i.e., the path $x_1\ldots x_sy_2\ldots y_t$.
For a $(u,v)$-path $P$, we write $x\preceq_P y$ if $\dist_P(u,x)\leq \dist_P(u,y)$, and $x\prec_P y$ if $x\preceq_P y$ and $x\neq y$. Respectively, for $xy,x'y'\in E(P)$, $xy\preceq_P x'y'$ if $x,y\preceq_P x',y'$. Notice that we always assume that the first vertex $u$ of $P$ is specified whenever we use this notation.

Let $G$ be a graph. For $u,v\in V(G)$, we say that $\{u,v\}$ is a \emph{diameter pair} if $\dist_G(u,v)=\diam(G)$.
A path $P$ in $G$ is a \emph{dominating path} if $V(P)$ is a \emph{dominating set} of $G$, i.e., each vertex of $G$ is either in $V(P)$ or adjacent to a vertex of $V(P)$. A pair of vertices $\{u,v\}$ is a \emph{dominating pair} if any $(u,v)$-path in $G$ is a dominating path.  Respectively, $\{u,v\}$ is a \emph{diameter dominating pair} if $\{u,v\}$ is both a diameter pair and a dominating pair.    
Corneil, Olariu and Stewart~\cite{CorneilOS97,CorneilOS99} proved that every connected AT-free graph has a diameter dominating pair.

\begin{lemma}[\cite{CorneilOS97,CorneilOS99}]\label{lem:dom_pair} 
Every connected AT-free graph has a diameter dominating pair, and such a pair can be found in $O(n^3)$ time. 
\end{lemma}

We need the following lemmas.

\begin{lemma}\label{lem:path-contr}
Let  $\{u,v\}$ be a diameter pair in a connected graph $G$, let $s\geq 2$ be an integer, and let $d=\diam(G)>s$. If a graph $H$ of diameter at most $s$ can be obtained from $G$ by contracting at set $S$ of at most $k=d-s$ edges, then $|S|=k$ and there is a $(u,v)$-path $P$ of length $d$ such that $S\subseteq E(P)$. 
\end{lemma}

\begin{proof}
Let $S$ be a set of at most $k$ edges of $G$ such that $H=G/S$ and let $\mathcal{W}=\{W(x)\mid x\in V(H)\}$ be the corresponding $H$-contraction structure. 
Because $\diam(H)\leq s$, $H$ has a path $Q=x_0\ldots x_t$ such that $u\in W(x_0)$, $v\in W(x_t)$ and $t\leq s$. By the definition of an $H$-contraction structure, for $i\in\{0,\ldots,t\}$, there are $y_i,z_i\in W(x_i)$ with the following properties: $y_0=u$, $z_t=v$, and
$z_{i-1}y_i\in E(G)$ for $i\in\{1,\ldots,t\}$, where $y_i=z_i$ is possible. Each $W(x_i)$ induces a connected subgraph of $G$. Moreover, the edges of $S$ in $G[W(x_i)]$ compose a spanning subgraph, i.e., $G_i'=(W(x_i), E(G[W(x_i)])\cap S)$ is connected.
Hence, for $i\in\{0,\ldots,t\}$, there is a  $(y_i,z_i)$-path $P_i$ in $G[W(x_i)]$ with $E(P_i)\subseteq S$.  
Denote by $P$ the $(u,v)$-path $P_0+z_0y_1+P_1+\ldots+z_{t-1}y_t+P_t$. Observe that $\sum_{i=0}^t|E(P_i)|\leq |S|\leq k$ and $P$ has length at most $k+t\leq k+s$. On the other hand, because $\dist_G(u,v)=d$, it holds that $P$ has length at least $d=k+s$. This implies that $t=s$ and $P$ has length exactly $k+t=d$. It also implies that $|S|=k$ and $S\subseteq E(P)$.
\end{proof}

\begin{lemma}\label{lem:AT-bounds}
Let $G$ be a connected AT-free graph and let $s\geq 2$ be an integer, and suppose that $\diam(G)>s$. If $k$ is the minimum number of edges that needs to be contracted in order to obtain a graph of diameter at most $s$ from $G$,
then $\diam(G)-s\leq k\leq \diam(G)-s+2$.
\end{lemma}

\begin{proof}
The bound $\diam(G)-s\leq k$ is a straightforward corollary of Lemma~\ref{lem:path-contr}.  By Lemma~\ref{lem:dom_pair}, 
$G$ has a diameter dominating pair $\{u,v\}$. Let $P$ be a $(u,v)$-path in $G$. This path has length $\diam(G)\geq \diam(G)-s+2$. Let $S$ be an arbitrary set of $\diam(G)-s+2$ edges of $P$.
Let $H=G/S$ and $Q=P/S$.  Because $P$ is a dominating path in $G$, $Q$ is a dominating  path in $H$. Since $Q$ has length at most $s-2$, it immediately implies that $\diam(H)\leq s$. Hence, $k\leq \diam(G)-s+2$. 
\end{proof}

Let us point out that the bounds in Lemma~\ref{lem:AT-bounds} are tight; if $G$ is a path, then it is sufficient to contract $\diam(G)-s$ edges to obtain a graph of diameter at most $s$. If $G$ is a  graph shown in Fig.~\ref{fig:ladder} for $r\geq 2$, then $\diam(G)=2r$ and it is necessary to contract $2r-s+2$ edges to obtain a graph of diameter $2\leq s\leq r$. 

\begin{figure}[ht]
\centering\scalebox{0.65}{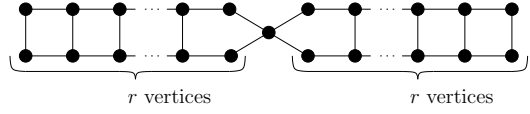}
\caption{A graph $G$ for which the upper bound in Lemma~\ref{lem:AT-bounds} is tight. 
\label{fig:ladder}}
\end{figure}

Let $G$ be a connected graph of diameter $d$. For a diameter pair $\{u,v\}$, let $X_u=\{x\in V(G)\mid \dist_G(u,x)=d\}$, $Y_u=\{x\in V(G)\mid \dist_G(u,x)=d-1\}$, $X_v=\{x\in V(G)\mid \dist_G(v,x)=d\}$, and $Y_v=\{x\in V(G)\mid \dist_G(v,x)=d-1\}$. We say that a path $Q=x_0\ldots x_k$ in $G$ is \emph{$(u,v)$-satisfying} if 
 \begin{itemize}
\item[i)] $Q$ is a subpath of some $(u,v)$-path $P$ of length $d$, and $x_0\prec_P x_k$;
\item[ii)] for any $z\in X_v$, $\dist_G(z,x_0)=\dist_G(u,x_0)$, and for any $z\in X_u$, $\dist_G(z,x_k)=\dist_G(v,x_k)$;
\item[iii)] for any  $z\in Y_v$,  $\dist_G(z,x_0)\leq \dist_G(u,x_0)$ or $\dist_G(z,x_1)\leq \dist_G(u,x_0)$, and for any  $z\in Y_u$,  $\dist_G(z,x_k)\leq \dist_G(v,x_k)$ or $\dist_G(z,x_{k-1})\leq \dist_G(v,x_k)$.
\end{itemize}

We need the following straightforward observation.

\begin{lemma}\label{lem:shift}
Let $\{u,v\}$ be a diameter pair in a connected graph $G$, and let $P$ be a $(u,v)$-path of length $d=\diam(G)$.
If for two vertices $w_1,w_2\in V(P)$ such that $w_1\prec_P w_2$, it holds that 
\begin{itemize}
\item[i)] $\dist_G(z,w_1)=\dist_G(u,w_1)$ for any $z\in X_v$, and 
\item[ii)] $\dist_G(z,w_1)\leq\dist_G(u,w_1)$ or $\dist_G(z,w_2)\leq \dist_G(u,w_2)-1$ for any $z\in Y_v$, 
\end{itemize}
then for any two vertices $w_1',w_2'\in V(P)$ such that $w_1'\prec_P w_2'$, $w_1\prec_P w_1'$, and $w_2\prec_P w_2'$, it holds that
\begin{itemize}
\item[i)] $\dist_G(z,w_1')=\dist_G(u,w_1')$ for any $z\in X_v$, and 
\item[ii)] $\dist_G(z,w_1')\leq\dist_G(u,w_1')$ or $\dist_G(z,w_2')\leq \dist_G(u,w_2')-1$ for any $z\in Y_v$. 
\end{itemize}
\end{lemma}

Now we need the structural results given in the following three lemmas. 

\begin{lemma}\label{lem:struct-diam}
Let $\{u,v\}$ be a diameter pair in a connected graph $G$, let $s\geq 2$ be an integer, and let $d=\diam(G)\geq s+3$. 
If a graph $H$ of diameter at most $s$ can be obtained from $G$ by contracting a set $S$ of at most $k=d-s$ edges, then
$G$ has a $(u,v)$-satisfying path of length $k$.
Moreover, if for any $(u,v)$-satisfying path  $Q=x_0\ldots x_k$, it holds that $x_0=u$ or $x_k=v$, then 
\begin{itemize}
\item[i)] $x_i$ is the unique vertex at distance $i$ from $u$ in $G$ for $i\in\{0,\ldots,k-2\}$ if $x_0=u$, and 
$x_{k-i}$ is the unique vertex at distance $i$ from $v$ in $G$ for $i\in\{0,\ldots,k-2\}$ if $x_k=v$; and
\item[ii)] $x_{i-1}x_i\in S$ for $i\in\{1,\ldots,k-2\}$ if $x_0=u$, and  
 $x_{i-1}x_i\in S$ for $i\in\{1,\ldots,k-2\}$ if $x_k=v$.
\end{itemize}
\end{lemma}

\begin{proof}
By Lemma~\ref{lem:path-contr}, there is a $(u,v)$-path $P$ of length $d$ in $G$ such that $S\subseteq E(P)$. Observe also that $|S|=k\geq 3$. 
Denote by $e_1,\ldots,e_k$, $e_1\preceq_P\ldots\preceq_P e_k$, the edges of $S$. Let $y_1$ be the end-vertex of $e_1$ closest to $u$ and let $y_2$ be the end-vertex of $e_2$ closest to $u$. Similarly, let $z_1$ be the end-vertex of $e_k$ closest to $v$ and let $z_2$ be the end-vertex of $e_{k-1}$ closest to $v$. Notice that $y_1\prec_P y_2\prec_P z_2\prec_P z_1$, and that $\dist_{G}(y_2,z_2)\geq k-2$. 

We show the following claims.

\medskip
\noindent
{\bf Claim 1.} {\it For any $z\in X_v$, $\dist_G(z,y_1)=\dist_G(u,y_1)$, and for any $z\in X_u$, $\dist_G(z,z_1)=\dist_G(v,z_1)$. }

\begin{proof}[Proof of Claim~1]
Clearly, it suffices to show that for any $z\in X_v$, $\dist_G(z,y_1)=\dist_G(u,y_1)$ as the second part of the claim is symmetric.  Because $P$ has  length $d=\diam(G)$, $\dist_G(z,y_1)\geq \dist_G(u,y_1)$ for $z\in X_v$. Suppose that there is $z\in X_v$ such that $\dist_G(z,y_1)>\dist_G(u,y_1)$. Then $e_1$ does not belong to any $(z,v)$-path of length $d$, contradicting Lemma~\ref{lem:path-contr} for the diameter pair $\{z,v\}$. 
\end{proof}

\medskip
\noindent
{\bf Claim 2.} {\it For any $z\in Y_v$, $\dist_G(z,y_1)\leq\dist_G(u,y_1)$ or $\dist_G(z,y_2)\leq \dist_G(u,y_2)-1$, and for any  $z\in Y_u$,  $\dist_G(z,z_1)\leq \dist_G(v,z_1)$ or $\dist_G(z,z_2)\leq \dist_G(v,z_2)-1$.}

\begin{proof}[Proof of Claim~2]
By symmetry, it is sufficient to show that for any $z\in Y_v$, $\dist_G(z,y_1)\leq\dist_G(u,y_1)$ or $\dist_G(z,y_2)\leq \dist_G(u,y_2)-1$.
Suppose that there is $z\in Y_v$ such that $\dist_G(z,y_1)>\dist_G(u,y_1)$ and  $\dist_G(z,y_2)> \dist_G(u,y_2)-1$. Because $k=d-s$, the contraction of $e_1,e_2$ should decrease the diameter of the graph by at least 2. But for $G'=G/\{e_1,e_2\}$,  we have that $\diam(G')\geq \dist_{G'}(z,v)=\dist_G(z,v)=d-1$; a contradiction.
\end{proof}

Observe that there exist four vertices $y_1',y_2',z_1',z_2'\in V(P)$ such that $y_1'\prec_P y_2'\prec_p z_2'\prec_P z_1'$, $y_1\preceq_P y_1'$, $y_2\preceq_P y_2'$, $z_1'\preceq_P z_1$, $z_2'\preceq_P z_2$, and $\dist_P(y_1',z_1')=k$. Let $x_0\ldots x_k$ be the $(y_1',z_1')$-subpath of $P$. Claims~1 and~2 together with Lemma~\ref{lem:shift} immediately imply that conditions ii) and iii) of the definition of a $(u,v)$-satisfying path are fulfilled.

To show the second part of the statement of the  lemma, observe that if  $\dist_P(u,z_2)\geq k$ and  $\dist_P(v,y_2)\geq k$, then $y_1',y_2',z_1',z_2'$ can always be chosen in such a way that $u\neq y_1'$ and $v\neq z_1'$. Hence,  if for any $(u,v)$-satisfying path $Q=x_0\ldots x_k$, $x_0=u$ or $x_k=v$, then 
either $\dist_P(u,z_2)\leq k-1$ and  $\dist_P(v,y_2)\leq k-1$. 

Suppose that $\dist_P(u,z_2)\leq k-1$ (the other case is symmetric). Assume that $G$ has a vertex $w\neq x_i$ at distance $i$ from $u$ in $G$ for some $i\in\{0,\ldots,k-2\}$. Observe that if $G'=G/S'$ for $S'=\{x_0x_1,\ldots,x_{k-2}x_{k-1}\}$, then $\dist_{G'}(w,v)\geq s+2$. Hence, $\dist_H(w,v)\geq s+1$ contradicting the condition that $\diam(H)=s$.  Therefore, i) follows. By the definition of $z_2$, $x_{i-1}x_i\in S$ for $i\in\{1,\ldots,k-2\}$ and ii) holds. 
\end{proof}

\begin{lemma}\label{lem:sufficient}
Let $\{u,v\}$ be a diameter dominating pair in a connected AT-free graph $G$ of diameter $d$, and let $s$ be an integer such that $2\leq s<d$.  
If $G$ has a $(u,v)$-satisfying path $Q=x_0\ldots x_k$ of length $k=d-s$ such that $u\neq x_0$ and $v\neq x_k$, then $\diam(G/E(Q))\leq s$.
\end{lemma}

\begin{proof}
Let $H=G/E(Q)$. Denote by $w$ the vertex of $H$ obtained from $x_0,\ldots,x_k$. By the definition, $Q$ is a subpath of  a $(u,v)$-path $P$ of length $d$. Because $\{u,v\}$ is a dominating pair, $P$ is a dominating path in $G$ and $P'=P/E(Q)$ is a dominating path in $H$. Let $P_1$ be the $(u,w)$-subpath of $P'$ and let $P_2$ be the $(w,v)$-subpath of $P'$.
Denote by $s_1$ the length of $P_1$ and denote by $s_2$ the length of $P_2$. Clearly, $s_1,s_2\geq 1$ and $s_1+s_2=s$.   
To show that $\diam(H)\leq s$, we have to prove that for any two vertices $y,z\in V(H)$, $\dist_H(y,z)\leq s$.

Suppose that $y$ is a vertex of a $P_1$ or is adjacent to a vertex of $P_1$, and $z$ is a vertex of a $P_2$ or is adjacent to a vertex of $P_2$. 
If $y=w$ or $y$ is adjacent to $w$ in $H$, then $\dist_H(y,w)\leq 1\leq s_1$.
Suppose that $y\neq w$ and is not adjacent to $w$. If $y\in X_v$, then $\dist_G(y,x_0)\leq s_1$. Hence, $\dist_H(y,w)\leq s_1$. If $y\in Y_v$, then   $\dist_G(y,x_0)\leq s_1$ or $\dist_G(y,x_1)\leq s_1$. Because $x_0x_1\in E(Q)$, it immediately implies that $\dist_H(y,w)\leq s_1$. If $y\notin X_v\cup Y_v$, then $y$ is not adjacent to $u$. Then $y$ is adjacent to some other vertex of $P_1$ or is a vertex of $P_1$ and, therefore, $\dist_H(y,w)\leq s_1$. We have that $\dist_H(y,w)\leq s_1$ in all cases. By the same arguments, $\dist_H(z,w)\leq s_2$. Therefore, $\dist_H(y,z)\leq s_1+s_2=s$.

Assume now that each of $y,z$ is a vertex of $P_1$ or is adjacent to a vertex of $P_1$ but $y,z\neq w$.  
If $y$ or $z$ is adjacent to $w$, we observe that this vertex is at distance at most $s_2$ from $w$ and apply the same arguments as above.
If  $y,z$ are not adjacent to $w$, then  we immediately obtain that $\dist_H(y,z)\leq \dist_G(y,z)\leq s$ because $s_1\leq s-1$. 
The case when  each of $y,z$ is a vertex of a $P_2$ or is adjacent to a vertex of $P_2$ but $y,z\neq w$ is symmetric. Hence, we again have that $\dist_H(y,z)\leq s$. 
\end{proof}

Now we analyze the case when the minimum number of contracted edges is $\diam(G)-s+1$.

\begin{lemma}\label{lem:del-edges}
Let $G$ be a connected AT-free graph of diameter $d$,  and let $s$ be an integer such that $2\leq s\leq d-2$. Suppose that the minimum number of edges that needs to be contracted in order to obtain a graph of diameter at most $s$ from $G$ is $k=d-s+1$. Then for any set $S\subseteq E(G)$ of size $k$ such that $\diam(G/S)\leq s$, there is a set $S'\subseteq S$ of size at most~$2$ such that $\diam(G/S')\geq d-|S'|+1$.  
\end{lemma}

\begin{proof}
Let $\{u,v\}$ be a diameter dominating pair in $G$. 
Let $S=\{e_1,\ldots,e_k\}\subseteq E(G)$ be a set of edges such that $\diam(H)\leq s$ for $H=G/S$. To obtain a contradiction, assume that for any  $S'\subseteq S$ of size at most~$2$, $\diam(G/S')\leq d-|S'|$.

If there is an edge $e_i\in S$ such that $e_i$ is not an edge of any $(u,v)$-path of length $d$, then $\diam(G/e_i)=d$, contradicting the assumption. Hence, every edge of $S$ is an edge of some 
$(u,v)$-path of length $d$. 

Suppose that for any two edges $e_i,e_j\in S$,  there is a $(u,v)$-path of length $d$ that contains them.
We show that there is a $(u,v)$-path $P$ of length $d$ such that $S\subseteq E(P)$.

Let $e_i=a_ib_i$ and assume that $\dist_G(u,a_i)<\dist_G(u,b_i)$ for $i\in\{1,\ldots,k\}$.
If $\dist_G(u,a_i)=\dist_G(u,a_j)$ for some distinct $i,j\in\{1,\ldots,k\}$, then any $(u,v)$-path that contains $e_i,e_j$ has length at least $d+1$. Hence, we can assume without loss of generality that $\dist_G(u,a_1)<\ldots<\dist_G(u,a_k)$. Let $r\leq k$ be the maximum integer such that there is a $(u,v)$-path $P$ that contains $e_1,\ldots,e_r$. Notice that $r\geq 2$. If $r=k$, our claim holds. Let $r<k$. Then let $P_1$ be the $(u,b_r)$-subpath of $P$ and let $P_2$ be the $(b_r,v)$-subpath of a $(u,v)$-path of length  $d$ containing $e_r,e_{r+1}$. Because $\dist_G(u,a_r)<\dist_G(u,a_{r+1})$, $P_2$ contains $e_{r+1}$. It remains to observe that the path $P_1+P_2$ contains $e_1,\ldots,e_{r+1}$ and has length $d$, contradicting the maximality of $r$.

Recall that $|S|=k\geq 3$ and $e_1\preceq_P\ldots\preceq_P e_k$.  Let $y_1$ be the end-vertex of $e_1$ closest to $u$ and let $y_2$ be the end-vertex of $e_2$ closest to $u$. Similarly, let $z_1$ be the end-vertex of $e_k$ closest to $v$ and let $z_2$ be the end-vertex of $e_{k-1}$ closest to $v$. Notice that $y_1\prec_P y_2\prec_P z_2\prec_P z_1$. Notice also that $\dist_{G}(y_2,z_2)\geq k-2$. We show the following claims.

\medskip
\noindent
{\bf Claim 1.} {\it For any $z\in X_v$, $\dist_G(z,y_1)=\dist_G(u,y_1)$, and for any $z\in X_u$, $\dist_G(z,z_1)=\dist_G(v,z_1)$. }

\begin{proof}[Proof of Claim~1]
We only show that for any $z\in X_v$, $\dist_G(z,y_1)=\dist_G(u,y_1)$; the second part of the claim follows by symmetry.  
Because $P$ has  length $d=\diam(G)$, $\dist_G(z,y_1)\geq \dist_G(u,y_1)$ for $z\in X_v$. Suppose that there is $z\in X_v$ such that $\dist_G(z,y_1)>\dist_G(u,y_1)$. Then contracting $e_1$ does not decrease the diameter because $\dist_{G'}(z,v)=d$ for $G'=G/e_1$; a contradiction.
\renewcommand{\qedsymbol}{$\diamond$}
\end{proof}

\medskip
\noindent
{\bf Claim 2.} {\it For any $z\in Y_v$, $\dist_G(z,y_1)\leq\dist_G(u,y_1)$ or $\dist_G(z,y_2)\leq \dist_G(u,y_2)-1$, and for any  $z\in Y_u$,  $\dist_G(z,z_1)\leq \dist_G(v,z_1)$ or $\dist_G(z,z_2)\leq \dist_G(v,z_2)-1$.}

\begin{proof}[Proof of Claim~2]
By symmetry, it is sufficient to show that for any $z\in Y_v$, $\dist_G(z,y_1)\leq\dist_G(u,y_1)$ or $\dist_G(z,y_2)\leq \dist_G(u,y_2)-1$.
Suppose that there is $z\in Y_v$ such that $\dist_G(z,y_1)>\dist_G(u,y_1)$ and  $\dist_G(z,y_2)> \dist_G(u,y_2)-1$. Let $S'=\{e_1,e_2\}$ and $G'=G/S'$. We have that $\dist_{G'}(z,v)=d-1$, i.e., $\diam(G')\geq d-1>\diam(G)-2$; a contradiction.
\renewcommand{\qedsymbol}{$\diamond$}
\end{proof}

\medskip
Observe that there exist four vertices
 $y_1',y_2',z_1',z_2'\in V(P)$ such that $y_1'\prec_P y_2'\prec_p z_2'\prec_P z_1'$, $y_1\preceq_P y_1'$,
$y_2\preceq_P y_2'$, $z_1'\preceq_P z_1$, $z_2'\preceq_P z_2$, and $\dist_P(y_1',z_1')=k-1$. 
Moreover, because $\dist_{G}(y_2,z_2)\geq k-2$, we can select $y_1'\neq u$ and $z_1'\neq v$.
Let $Q=x_0\ldots x_{k-1}$ be the $(y_1',z_1')$-subpath of $P$. Claims~1 and 2 together with Lemma~~\ref{lem:shift} immediately imply that $Q$ is an $(u,v)$-
satisfying path of length $k-1=d-s$. Moreover,  $u\neq x_0$ and $v\neq x_{k-1}$. By Lemma~\ref{lem:sufficient}, contracting $k-1$ edges in $Q$ yields  a graph of diameter at most $s$. But this contradicts the condition that the minimum number of edges that needs to be contracted in order to obtain a graph of diameter at most $s$ from $G$ is $k=d-s+1$.

It remains to consider the case when there are two distinct edges $e_i,e_j\in S$ such that any $(u,v)$-path of length $d$ in $G$ does not contain $e_i$ or $e_j$. Let $S'=\{e_i,e_j\}$ and observe that $\diam(G/S')\geq d-1$. This contradicts the assumption that $\diam(G/S')\leq d-|S'|$.  
\end{proof}

Now we are ready to prove the main result of the section.

\begin{theorem}\label{thm:cc}
For any $s\geq 2$, the {\sc $s$-Club Contraction} problem can be solved in $O(m^4n^3)$ time on AT-free graphs. This result holds even if $s$ is given as a part of the input.
\end{theorem}

\begin{proof}
Let $(G,k)$ be an instance of {\sc $s$-Club Contraction}. If $G$ is disconnected, then $G$ cannot be contracted to a graph of finite diameter.  
Notice that we can compute $d=\diam(G)$ in time $O(n^3)$.
If $d\leq s$, the problem is trivial. 
From now we assume that $G$ is connected and $d>s$. 
By Lemma~\ref{lem:AT-bounds}, if $k<d-s$, then we have a no-instance of the problem, and if $k\geq  d-s+2$, then we have a yes-instance.  
Hence, we can assume that $d-s\leq k\leq d-s+1$. 

Suppose that $k=d-s$.  
If $k\leq 2$, we solve the problem by brute force in time $O(m^2n^3)$ by checking all possible choices of $k$ edges. 
Let $k\geq 3$.  Using Lemma~\ref{lem:dom_pair}, we find a diameter dominating pair $\{u,v\}$.
By considering all possible choices for the edges $x_0x_1$ and $x_{k-1}x_k$, we check for the existence of $(u,v)$-satisfying paths $Q=x_0\ldots x_k$ in time $O(m^2(n+m))$. 
If such a path does not exist, then we have a no-answer due to Lemma~\ref{lem:struct-diam}. If there is a $(u,v)$-satisfying path $Q=x_0\ldots x_k$ with $u\neq x_0$ and $v\neq x_k$, then by Lemma~\ref{lem:sufficient}, by contracting the edges of $Q$ we obtain a graph of diameter at most $s$. If we have only a $(u,v)$-satisfying path $Q=x_0\ldots x_k$ with $u=x_0$ or $v=x_k$, then we apply Lemma~\ref{lem:struct-diam}. 
If $u=x_0$, then we check whether 
$x_i$ is the unique vertex at distance $i$ from $u$ in $G$ for $i\in\{0,\ldots,k-2\}$.
If this is not the case, then Lemma~\ref{lem:struct-diam} guarantees that we have a no-answer. Otherwise,
we contract   
the edges $x_{i-1}x_i\in S$ for $i\in\{1,\ldots,k-2\}$. We then try all possible choices for two remaining edges and check whether we obtain a graph of diameter at most $s$.  We use the symmetrical arguments if $v=x_k$. Observe that the entire procedure can be performed in $O(m^2n^3)$ time. 
 
Assume now that $k=d-s+1$. First, we solve the problem for the instance $(G,k-1)$ in $O(m^2n^3)$ time using the above procedure. Clearly, if we have a yes-instance, $(G,k)$ is a yes-instance as well. Suppose that it is not sufficient to contract $k-1$ edges to obtain a graph of diameter at most $s$. Then we apply Lemma~\ref{lem:del-edges} and check, in $O(m^2)$ time, all sets of edges $S'$ of size at most~$2$. For each set, we first check in time $O(n^2)$ whether $\diam(G/S')=d-|S'|+1$. If so, then we solve the problem for the instance $(G',k')$ where $G'=G/S'$ and $k'=k-|S'|$. Because $k'=\diam(G')-s$, this can be done in $O(m^2n^3)$ time as described earlier. If we have a yes-answer for one of the instances, then we return a yes-answer. Otherwise, we have a no-answer.

It remains to observe that the total running time is $O(m^4n^3)$.
\end{proof}

\subsection{$s$-Club Contraction for chordal graphs}\label{sec:chord}

We now turn our attention to chordal graphs.

\begin{theorem}
\label{t-split}
For any $s\geq 2$, the {\sc $s$-Club Contraction} problem on chordal graphs is $\NP$-complete as well as $\W[2]$-hard when parameterized by $k$. Moreover, {\sc $2$-Club Contraction} is $\NP$-complete and $\W[2]$-hard when parameterized by~$k$ even on split graphs.
\end{theorem}

\begin{proof}
First, we show hardness for {\sc $2$-Club Contraction} on split graphs.

We reduce from the {\sc Hitting Set} problem, which takes as input a finite set $U$, a collection ${\cal S}$ of subsets of $U$, and an integer $k$, and asks whether there exists a subset $U'\subseteq U$ of size at most $k$ such that $U'$ contains at least one element from each subset in ${\cal S}$; such a subset $U'$ is called a {\em hitting set} of size at most $k$. This problem is well-known to be \NP-complete~\cite{GareyJ79} as well as $\W[2]$-hard when parameterized by $k$~\cite{DowneyF13}. 

Given an instance $(U,{\cal S},k)$ of the {\sc Hitting Set} problem with $U=\{u_1,\ldots,u_n\}$ and ${\cal S}=\{S_1,\ldots, S_m\}$, we create a split graph $G$ as follows. We start by creating a vertex $u_i$ for each $u_i\in U$, and we make all these vertices into a clique that we denote by $U$. For every $S_j\in {\cal S}$, we create 
$2k+1$ vertices $S_j^{1},\ldots,S_j^{2k+1}$ that are made adjacent to vertex $u_i$ if and only if $u_i\in S_j$, for each $i\in \{1,\ldots,n\}$. We then add a vertex $x$ that is made adjacent to every vertex in $U$, as well as $2k+1$ vertices $y_1,\ldots,y_{k+1}$ that are made adjacent to $x$ only. This completes the construction of $G$. 

Note that the vertex set of $G$ can be partitioned into a clique $U\cup \{x\}$ and an independent set $V(G)\setminus (U\cup \{x\})$, so $G$ is a split graph. Also observe that the diameter of $G$ is~$3$. Hence, in order to finish the proof, it suffices to show that $G$ is $k$-contractible to a graph of diameter at most~$2$ if and only if $(U,{\cal S},k)$ is a yes-instance of {\sc Hitting Set}.

First suppose there exists a set $U'\subseteq U$ with $|U'|\leq k$ such that $U'\cap S_j\neq \emptyset$ for every $j\in \{1,\ldots,m\}$. Consider the corresponding $|U'|$ vertices in $G$, 
contract the edges that join the vertices of $U'$ and $x$, and denote by $w$ the obtained vertex. 
Let $H$ denote the resulting graph. The fact that $U'$ is a hitting set implies that in $H$, vertex $S_j^{p}$ is adjacent to $w$ for every $j\in \{1,\ldots,n\}$ and $p\in \{1,\ldots,2k+1\}$. This means that $w$ is a universal vertex in~$H$, implying that $H$ has diameter at most~$2$.
 
For the reverse direction, suppose there exists a graph $H$ of diameter at most~$2$ such that $G$ is $k$-contractible to $H$. Let ${\cal W}$ be an $H$-contraction structure of $G$. Due to Lemma~\ref{l-big}, we know that for each $j\in \{1,\ldots,m\}$, one of the vertices $S_j^{1},\ldots,S_j^{2k+1}$ forms a 
singleton, and the same holds for one of the vertices $y_1,\ldots,y_{2k+1}$. Without loss of generality, assume that each of the vertices $S_1^1,S_2^1\ldots,S_m^1,y_1$ forms a 
singleton. In particular, this means that every edge incident with these vertices is a witness edge, that is, an edge whose endpoints belong to two different bags. Consequently, there must be a bag $W\in {\cal W}$ that is adjacent to each of the vertices $S_1^1,S_2^1\ldots,S_m^1,y_1$. Due to Lemma~\ref{l-big}, this bag $W$ contains at most $k+1$ vertices. It is clear that $x\in W$, as $x$ is the unique neighbor of $y_1$ in $H$. Let $U'=W\setminus \{x\}$. Since none of the vertices $S_1^1,S_2^1\ldots,S_m^1$ is adjacent to $x$, each of 
them is adjacent to at least one vertex in $U'$. We conclude that $U'$ is a hitting set of size $|W|-1\leq k$.

To show that {\sc $3$-Club Contraction} is hard, we modify the above construction as follows. Instead of adding $y_1,\ldots,y_{2k+1}$ adjacent to $x$, we crate $2k+1$ vertices $z_1,\ldots,z_{2k+1}$ and make the set $\{z_1,\ldots,z_{2k+1},x\}$ into a clique. Now we construct $y_1,\ldots,y_{2k+1}$ and make $y_i$ adjacent to $z_i$ for $i\in\{1,\ldots,2k+1\}$. It is easy to see that the obtained graph $G$ is chordal, and by the same arguments as above, we have that $G$  is $k$-contractible to a graph of diameter at most~$3$ if and only if $(U,{\cal S},k)$ is a yes-instance of {\sc Hitting Set}. 

Finally, let $s\geq 4$. Consider a graph $G$ and denote by $G'$ the graph obtained from $G$ by adding $k+1$ pendant vertices adjacent to  $v$ for each vertex $v$  of $G$. It is straightforward to observe that $G'$ is $k$-contractible to a graph of diameter at most~$s$ if and only if $G$ is $k$-contractible to a graph of diameter at most~$s-2$. Clearly, if $G$ is chordal, then $G'$ is chordal as well.
As we already proved that 
 {\sc $s$-Club Contraction} is $\NP$-complete as well as $\W[2]$-hard when parameterized by $k$ for chordal graphs for $s\in \{2,3\}$, this observation immediately implies that 
 {\sc $s$-Club Contraction} is $\NP$-complete and $\W[2]$-hard  for chordal graphs for every fixed $s\geq 2$.
\end{proof}

\section{Concluding Remarks}
\label{sec:conclusions}

In Section~\ref{sec:main}, we showed that the {\sc Hadwiger Number} problem can be solved in polynomial time on cographs and on bipartite permutation graphs, respectively. A natural question is how far the results in those two sections can be extended to larger graph classes. An easy reduction from the {\sc Hadwiger Number} problem on general graphs, involving subdividing every edge of the input graph exactly once, implies that the problem is \NP-complete on bipartite graphs. Since bipartite permutation graphs form exactly the intersection of bipartite graphs and permutation graphs, and the class of permutation graphs properly contains the class of cographs, our results naturally raise the question whether the {\sc Hadwiger Number} problem can be solved in polynomial time on permutation graphs. We leave this as an open question. We point out that the problem is \NP-complete on co-comparability graphs, a well-known superclass of permutation graphs, due to Theorem~\ref{t-co-bipartite} and the fact that co-bipartite graphs form a subclass of co-comparability graphs.

In Section~\ref{sec:diam}, we proved that {\sc $s$-Club Contraction} is polynomial-time solvable on AT-free graphs for $s\geq 2$.
An interesting direction for further research is to identify other non-trivial graph classes for which the {\sc $s$-Club Contraction} problem is polynomial-time solvable (or fixed-parameter tractable when parameterized by~$k$) for all values of $s\geq 2$.


\end{document}